\documentclass[12pt,a4paper,twoside]{article}
\usepackage{amsmath,amsthm,latexsym}

\usepackage{amsfonts,amssymb}

\newtheorem{definition}{Definition}
\newtheorem{theorem}[definition]{Theorem}
\newtheorem{proposition}[definition]{Proposition}
\newtheorem{lemma}[definition]{Lemma}
\newtheorem{corollary}[definition]{Corollary}
\newtheorem{rem}[definition]{Remark}
\newenvironment{remark}{\begin{rem}  \rm }{\end{rem}}
\newtheorem{rems}[definition]{Remarks}

\newtheorem{example}[definition]{Example}

\newcommand{\initsp}{\mathsf{h}}
\newcommand\unit{\hbox{\rm 1\kern-2.8truept l}}

\newcommand\modd{\kern-8pt\mod}
\newcommand\modp{\kern-5pt\mod}

\title{Stationary states of weak coupling limit type \\ Markov generators and quantum transport models  }
\begin{document}
\author{\'Alvaro Hern\'andez-Cervantes$^1$ and Roberto Quezada$^2$}
%\date{}
\maketitle
\begin{quote}
\footnotesize{

Universidad Aut\'onoma Metropolitana, Iztapalapa Campus, Av. San Rafael Atlixco 186, Col. Vicentina, 09340 Iztapalapa D.F., M\'exico. \\ 
E-mail: \texttt{ahc\_89@hotmail.com}$^{1}$,  \texttt{roqb@xanum.uam.mx}$^{2}$
}
\end{quote}
\begin{abstract}
We prove that every stationary state in the annihilator of all Kraus operators of a weak coupling limit type Markov generator consists of two pieces, one of them supported on the interaction-free subspace and the second one on its orthogonal complement. In particular, we apply the previous result to describe in detail the structure of a slightly modified quantum transport model due to Arefeva, Kozyrev and Volovich (modified AKV's model) studied first in Ref.\cite{gggq}, in terms of generalized annihilation and creation operators. 
\end{abstract}

\section{Introduction}
The main aim of this work is two-fold. We describe the structure of stationary states in the annihilator of all Kraus (or noise) operators of the class of weak coupling limit type Markov generators (WCLT generators) introduced in Ref.\cite{AccardiFQ} and, in particular, we describe in detail the structure of the stationary states of the modified transport AKV's model, studied first in Ref.\cite{gggq}, in terms of generalized annihilation and creation operators. 

A natural space to search for stationary states of WCLT generators is the annihilator of all Kraus or noise operators defined as \[Ann(D):=\{ \rho: \textrm{tr}(\rho D_{\omega})=0, \; \; \textrm{for every Bohr frequency} \; \; \omega\}\] that contains (sometimes properly) the commutant $\{H\}'$ of the reference Hamiltonian. We show that stationary states in $Ann(D)$ consist  of a piece supported on the intersection of the kernels of all Kraus operators and its adjoints, denoted by $W_{D}$ and called interaction-free subspace, and another part supported on $W_{D}^{\perp}$. While the part supported on $W_{D}$ belongs to the fixed points  sub-algebra of the semigroup, the second is much more interesting and include detailed balance as well as non-detailed balance but in the class of local detailed balance stationary states. 

For the modified AKV's model, stationary states supported on $W_{D}^{\perp}$, exhi\-bit a very interesting structure described explicitly by using an o\-pe\-ra\-tor $Z$, called interference operator in Ref.\cite{gggq}, and its adjoint $Z^{*}$. Operator $Z$ maps the support subspace $P_{2}(\initsp)$ of the degenerate second spectral projection of the reference Hamiltonian $H$, into the support subspace $P_{3}(\initsp)$ of the third spectral projection. Operators $Z$ and $Z^{*}$ perform transitions between $P_{2}(\initsp)$ and $P_{3}(\initsp)$, similar to birth and death transitions in classical stochastic processes or creation and annihilation operators in the quantum setting. We show in Theorem \ref{approach-equilibrium} that any stationary state of the modified AKV's model, is a mixture of a state $\sigma$ supported on a subspace of $P_{2}(\initsp)$ and its conjugation by $Z$, $Z\sigma Z^{*}$, that is supported on a subspace of $P_{3}(\initsp)$. It turns out that the slightly modification of AKV's transport model performed in Ref.\cite{gggq} yields another quantum transport model, indeed: the total probability mass of any initial state of the modified transport model can be redistributed as $t\to \infty$, choosing appropriate values of the parameters of the model (the $\Gamma$'s), so that most of the total probability of the final state is concentrated on the portion supported on $P_{3}$, even though the support of the final state is not contained in $P_{3}(\initsp)$ (there is a non-zero trace on $P_{2}(\initsp)$), see Remark \ref{transport}. Moreover, the evolution increases the energy of the small system when it is initially in a state supported on a distinguished subspace of $P_{2}(\initsp)$, so that when $t\to\infty$ there is an energy gain proportional to $(N-1)=\textrm{dim}P_{2}(\initsp)$, see Remark \ref{energy-gain}.  

The paper is organized as follows. After some preliminaries, in Section \ref{states-annihilator} we describe the structure of stationary state WCLT Markov generators as a convex combination of a state supported on the interaction-free subspace $W_{D}$ and another state supported on the orthogonal complement $W_{D}^{\perp}$. In Section \ref{AKV-model} we apply the results of the previous section to study the stationary state of the modified AKV's model by a method considerable shorter than the one used in Ref.\cite{gggq} In contradistinction with the previous reference, where the assumption that any stationary state belongs to $\{H\}'$ is used, we start with the computation of subharmonic projections and prove that any stationary state belongs to $\{H\}'$. Moreover we study the approach to equilibrium property for this semigroup and characterize the attraction domain of any stationary state. 
  
\section{Preliminaries}\label{preliminaries}
Following Ref.\cite{AccardiFQ}, we call Markov generator (or simply generator) any formal expression with the GKSL structure from which one can construct a Quantum Markov Semigroup. In particular we shall consider Markov generators of weak coupling limit type ${\mathcal L}$, associated with a positive self-adjoint operator (reference Hamiltonian) with discrete spectral decomposition 
\begin{equation}\label{discr-sa-H}
H=\sum_{\varepsilon_{m}\in{\rm Sp}(H)}
\varepsilon_{m} P_{\varepsilon_m}
\end{equation} Where $\varepsilon_{m}$ are the eigenvalues and $P_{\varepsilon_m}$ is the spectral projection on the eigenspace associated with $\varepsilon_{m}$. 

A weak coupling limit type Markov generator has the structure 
\begin{eqnarray}
{\mathcal L} = \sum_{\omega\in B_{+}} {\mathcal L}_{\omega} 
\end{eqnarray} 

The subset of ordered pairs of eigenvalues (Arveson spectrum) 
\begin{equation}\label{df-hat-B+}  
B_{+} :=\{(\varepsilon_{n} , \varepsilon_{m})  \ : \ 
\varepsilon_{n} - \varepsilon_{m} >0 \}
\end{equation} is the set of Bohr frequencies and 
\begin{equation}\label{df-hat-B+01}
{B}_{+,\omega}:=\{(\varepsilon_{n} ,\varepsilon_{m})\in{B}_{+}  \ : \ 
\varepsilon_{n} -\varepsilon_{m} =\omega\}
\end{equation} is the set of pairs of eigenvalues associated with the frequency $\omega$.

For every Bohr frequency $\omega$ and $x\in{\mathcal B}(\initsp)$ the corresponding generator ${\mathcal L}_{\omega}$ has the Gorini-Kossakowski-Sudarshan and Lindblad (GKSL) canonical form   

\begin{eqnarray}
\begin{aligned}
& \mathcal{L}_\omega(x)  := i[\Delta_\omega,x] \\ & 
 - \Gamma_{-,\omega}
\left(\frac{1}{2}\{D^{*}_\omega D_\omega,x\}
-D^{*}_\omega x D_\omega\right) - \Gamma_{+,\omega}
\left(\frac{1}{2}\{D_\omega D^{*}_\omega,x\}
-D_\omega x D^{*}_\omega\right)
\end{aligned}
\end{eqnarray} with partial interaction (or Kraus) operators 
\begin{equation}\label{df-Lomega1}       
D_\omega  := \sum_{(\varepsilon_n,\varepsilon_m)\in {B}_{+,\omega}}
P_{\epsilon_m}D P_{\epsilon_n}, 
\end{equation} with  $D\in{\mathcal B}(\initsp)$ an interaction operator.
The effective Hamiltonian 
\begin{equation}\label{prp-Domega2}       
\Delta_\omega =\Delta^*_\omega\in \{H\}' 
\end{equation}  
\begin{equation}\label{the-Gamma's}
\Gamma_{+ ,\omega} =c_{\omega}\frac{1}{e^{\beta(\omega)\omega} -1}
\quad ;\quad 
\Gamma_{- ,\omega} =c_{\omega}\frac{e^{\beta(\omega)\omega}}{e^{\beta(\omega)\omega}-1}
\qquad ;\qquad c_{\omega}\geq 0
\end{equation}
$\beta$ a regular enough function (inverse temperature) with 
\begin{equation}\label{beta(om)>0}    
\beta(\omega) >0 \qquad;\qquad \forall \omega\in B_+ 
\end{equation}
The predual generator has the form \[{\mathcal L}_{*}= \sum_{\omega\in B_{+}}{\mathcal L}_{\omega *}\] where for each Bohr frequency $\omega$, 
\begin{eqnarray}\label{df-Lomega0a} 
\begin{aligned} 
{\mathcal L}_{\omega *}(\rho) &  
=\left(-\frac{\Gamma_{-,\omega}}{2}
\{D^{*}_\omega D_\omega,\rho\}
+ \Gamma_{+,\omega}D^{*}_\omega\rho D_\omega\right) \\ &
+ \left(-\frac{\Gamma_{+,\omega}}{2}
\{D_\omega D^{*}_\omega,\rho\}
+\Gamma_{-,\omega}D_\omega\rho D^{*}_\omega\right) -i[\Delta_{\omega}, \rho]
\end{aligned}
\end{eqnarray}

From now on we will write simply $P_{n}$ instead $P_{\epsilon_n}$ whenever there is no confusion. 

\section{Stationary states in the annihilator of all Kraus operators}\label{states-annihilator}
In this section we study stationary states of WCLT generators that belong to the annihilator of all Kraus operators defined as \[Ann(D):=\{ \rho: \textrm{tr}(\rho D_{\omega})=0, \; \; \textrm{for every Bohr frequency} \; \; \omega\}\] Notice that \[\{H\}'' \subseteq \{H\}'\] with equality if $H$ is non degenerate. Moreover, for the class of WCLT Markov generators the inequality \[\{H\}' \subset Ann(D)\] holds. Indeed, for any $\rho\in\{H\}'$, 
\[Tr(\rho \sum_{(\varepsilon_n,\varepsilon_m)\in {B}_{+,\omega}}
P_{m}D P_{n} )= \sum_{(\varepsilon_n,\varepsilon_m)\in {B}_{+,\omega}} Tr(P_{n} P_{m}\rho D)=0,\] since $n\neq m$.  
But one can find WCLT generators with invariant states $\rho \in Ann(D)\setminus\{H\}'$. See Corollary \ref{inv-notin-comm}, below.

\subsection{Stationary states supported on the interaction-free subspace}
We start by describing the simplest class of stationary states in the annihilator of all Kraus operators, which are characterized by the stronger condition $\rho D_{\omega} =0=D_{\omega}\rho$ for all $\omega$.  The support subspace of these stationary states, defined as the closure of its range, is contained in the intersection of the kernels of all Kraus operators; moreover, these stationary states do not dependent on the coefficients $\Gamma_{\pm, \omega}$.

\begin{definition}
\begin{itemize}
\item[(i)]We call singular any subset $\Theta\subseteq {\mathcal B}(\initsp)$ satisfying the condition  
\[\ker\Theta:=\displaystyle\bigcap_{\theta\in \Theta}\ker\theta \neq \{0 \}\] 
\item[(ii)] The \textit{interaction-free subspace} 
%(also called decoherence-free subspace\cite{li12,LCW}), 
is defined by \[W_{D}:=\displaystyle\bigcap_{\omega\in\mathbf{B}_{+}}\left(\ker D_{\omega}\cap \ker D_{\omega}^{\ast}\right)\]
\end{itemize}
\end{definition} 

\begin{remark}
The relation of the interaction-free subspace $W_{D}$ with the notion of decoherence for quantum Markov semigroups\cite{BlOl03,BlOl06,CSU12} and in particular with the decoherence-free subspace\cite{li12,LCW,AgFfRr} will be discussed in forthcoming work.
\end{remark}

The proof of the following proposition is straightforward.

\begin{proposition}\label{theorem-general-hamiltonian-10} 
Let ${\mathcal L}$ be a WCLT generator and $D\in\mathcal{B}(\initsp)$ an interaction operator. Then, 
\begin{itemize}
\item[(i)] $W_{D}\neq \{0\}$ if and only if the subset \[\Theta:=\{P_{l}DP_{k}, P_{k}D^{\ast}P_{l} : l<k\}\] is singular. 
\item[(ii)] $0\neq u\in W_D$ if and only if $u=\sum_{l}P_{l}u= \sum_{l} u_{l}$, with \[u_{l}\in V_{l}:=\displaystyle\Big(\bigcap_{k<l}\ker(P_{k}DP_{l})\displaystyle\bigcap_{k>l}\ker(P_{k}D^{\ast}P_{l})\Big)\cap P_{l}({\initsp})\]  for every $l$. 
\item[(iii)] $V_{k}\perp V_{l}$ if $k\neq l$ and $W_D=\bigoplus_{l} V_{l}$.
\end{itemize}
\end{proposition}

\begin{theorem}\label{theorem-general-hamiltonian-11}
With the notations in the above proposition, let $\mathcal{L}_{*}:=\mathcal{L}_{*0}+ i \delta_{\Delta}$ be a Markov generator of WCLT with associated reference Hamiltonian $H\in\mathcal{B}(\initsp)$.  
Let $D\in\mathcal{B}(\initsp)$ 
be an interaction operator satisfying $W_{D} \neq \{0\}$. Then for every $u_{j}\in W_D$, $\|u_{j}\|=1$ and $\sum_{j} \lambda_{j}=1$, the state  
\begin{eqnarray}\label{ecu-state}
\begin{aligned}
\rho:=\sum_{j}\lambda_{j} |u_{j}\rangle \langle u_{j}|
\end{aligned}
\end{eqnarray}  satisfies ${\mathcal L_{*0}(\rho)=0}$. 
Moreover, if $\delta_{\Delta_{\omega}}(\rho)= [\Delta_{\omega}, \rho]$, where  
\begin{eqnarray}\label{ecu-Delta}
\begin{aligned}
\Delta_{\omega}= \zeta_{-,\omega}D_{\omega}^{\ast}D_{\omega}+\zeta_{+,\omega}D_{\omega}D_{\omega}^{\ast}
\end{aligned}
\end{eqnarray}
and $\zeta_{-,\omega},\zeta_{+,\omega}\in\mathbb{R}$, then $\rho$ is an ${\mathcal L}_{*}$-invariant state. 
\end{theorem}
\begin{proof}
The generator ${\mathcal L}_{*0}$ is a closed operator on the Banach space of finite-trace operators on $\initsp$. Being $\rho$ the limit of a sequence of finite-rank approximations $\rho_{N}\in\textrm{dom}({\mathcal L}_{*0})$, the identity ${\mathcal L}_{*0}(\rho_{N})=0$ for all $n\geq 0$, implies that $\rho\in{\textrm{dom}}({\mathcal L}_{*0})$ and ${\mathcal L}_{*0}(\rho)=0$.  Moreover, condition $u_{j}\in W_{D}$ for every $j$, implies that $\Delta_{\omega}$ commutes with $\rho$. 
\end{proof}

\begin{remark}
\begin{itemize}
\item [(i)] Notice that condition $u\in W_{D}$, implies that the pure state $|u\rangle\langle u|$ 
is an extremal stationary state. Moreover this operator belongs to the fixed points ${\mathcal F}({\mathcal T})$ of the direct semigroup $({\mathcal T}_{t})_{t\geq 0}$.

\item[(ii)] In the finite-dimensional case, one can write $\mathcal{B}(\mathbb{C}^{d})=\{H\}^{'}\oplus(\{H\}^{'})^{\perp}$, with respect to the Hilbert-Schmidt inner product. Then for every $D\in \mathcal{B}(\mathbb{C}^{d})$, there exist unique $D_{1}\in\{H\}^{'}$ and $D_{2}\in(\{H\}^{'})^{\perp}$ such that $D=D_{1}+D_{2}$. Hence, for each Bohr frequency $\omega:=\epsilon_{n}-\epsilon_{m}>0$ we have \begin{eqnarray*}
\begin{aligned}
D_{\omega}&:=\displaystyle\sum_{(\epsilon_{n},\epsilon_{m})\in \mathbf{B}_{+,\omega}}P_{m}DP_{n}=\displaystyle\sum_{(\epsilon_{n},\epsilon_{m})\in \mathbf{B}_{+,\omega}}P_{m}(D_{1}+D_{2})P_{n} \\
&=\displaystyle\sum_{(\epsilon_{n},\epsilon_{m})\in \mathbf{B}_{+,\omega}}P_{m}D_{1}P_{n}+\displaystyle\sum_{(\epsilon_{n},\epsilon_{m})\in \mathbf{B}_{+,\omega}}P_{m}D_{2}P_{n} \\ 
&= \displaystyle\sum_{(\epsilon_{n},\epsilon_{m})\in \mathbf{B}_{+,\omega}}P_{m}D_{2}P_{n}
\end{aligned}
\end{eqnarray*}
This implies that the interesting interaction operators are those belonging to $(\{H\}^{'})^{\perp}$.
\end{itemize}
\end{remark}

Our next theorem gives a characterization of states supported on $W_{D}$ in terms of its product with Kraus operators. 

\begin{theorem}\label{inv-WD}
An invariant state $\rho$ has the form (\ref{ecu-state}) for some $u_{j} \in W_{D}$ and $\lambda_{j}>0$ if and only if 
\begin{eqnarray}\label{rho-D-omega}
\rho D_{\omega} = 0 = D_{\omega} \rho
\end{eqnarray} for every Bohr frequency $\omega$. 
\end{theorem}
\begin{proof}
If $\rho$ has the form (\ref{ecu-state}) then it immediately follows that $\rho D_{\omega}=0= D_{\omega} \rho$ since the support of $\rho$ is a subset of $W_{D}$. 

Conversely, let $\rho$  be a state, then $\rho = \sum_{k} \rho_{k} |u_{k}\rangle \langle u_{k}|$ for some complete orthonormal system $(u_{k})$ of $\initsp$ and $\rho_{k}>0$. Moreover, if for every Bohr frequency $\rho D_{\omega} = 0 = D_{\omega} \rho$ then $\sum_{k} \rho_{k} |D_{\omega} u_{k}\rangle \langle u_{k}| = D_{\omega} \rho =0$. Therefore for every $u_{l}$ we have that $0=\sum_{k} \rho_{k} |D_{\omega} u_{k}\rangle \langle u_{k}| u_{l} = \rho_{l} D_{\omega} u_{l}$, which implies that  $u_{l}\in\ker D_{\omega}$ for every frequency $\omega$. Likewise, one proofs that $u_{l}\in \ker D_{\omega}^{*}$ for every $\omega$. We can conclude that $u_{l}\in W_{D}$ for every $l$ and, hence, $\rho$ has the form (\ref{ecu-state}). This finishes the proof.
\end{proof}

\subsection{Local detailed balance stationary states}

The stationary states of the form (\ref{ecu-state}) in Theorem \ref{theorem-general-hamiltonian-11}, are the simplest invariant states of a WCLT generator but not all invariant states belong to this class. In the case of a degenerate reference Hamiltonian, the stationary states exhibit a block structure, with blocks that may depend on the $\Gamma$'s, see Section \ref{AKV-model} and examples in Subsection \ref{exam}. 

Given a state $\rho$ and a subspace $V$, by ''$\rho$ is supported on $V$" we mean that ''the support of $\rho$ is a subspace of $V$". We start with a generalization of the theorem in previous section. 

\begin{theorem}\label{structure-rho}
With the notations in Theorem \ref{theorem-general-hamiltonian-11}, let ${\mathcal L}_{*} = \mathcal{L}_{*0} + i\delta_{\Delta}$ be a WCLT generator associated with (possible degenerate) $H$. Then 
\begin{itemize}
\item[(i)] a state $\rho$ commutes with the orthogonal projection on the subspace $W_{D}$, $P_{W_{D}}$, if and only if $\rho$ has the structure 
\begin{eqnarray}\label{structure-rho-1}
\rho = P_{W_{D}}\rho P_{W_{D}} + P_{W_{D}^{\perp}} \rho P_{W_{D}^{\perp}}
\end{eqnarray} where $P_{W_{D}^{\perp}}$ is the orthogonal projection on $W_{D}^{\perp}$.
 
\item[(ii)] If moreover $\rho$ is invariant for ${\mathcal L}_{*}$, then  
\begin{eqnarray}\label{structure-rho-01}
\rho = \theta \rho_{W_{D}} + (1-\theta)\rho_{W_{D}^{\perp}}, \; \; \; \; \theta\in [0,1]
\end{eqnarray} with $\rho_{W_{D}}$ an invariant state supported on $W_{D}$ and $\rho_{W_{D}^{\perp}}$ an invariant state supported on $W_{D}^{\perp}$.
\end{itemize}
\end{theorem}
\begin{proof} We have that, $P_{W_{D}} + P_{W_{D}^{\perp}}=\unit$, then for any state $\rho$,   

\[P_{W_{D}} \rho + P_{W_{D}^{\perp}}\rho = \rho = \rho P_{W_{D}} + \rho P_{W_{D}^{\perp}}\] Hence, $\rho$ commutes with $P_{W_{D}}$ if and only if it commutes with $P_{W_{D}^{\perp}}$. 

The assumption $\rho$ commutes with $P_{W_{D}}$ and simple computations show that  
\[\rho = \big(P_{W_{D}} + P_{W_{D}^{\perp}}\big)\rho \big(P_{W_{D}} + P_{W_{D}^{\perp}}\big) = P_{W_{D}}\rho P_{W_{D}} + P_{W_{D}^{\perp}} \rho P_{W_{D}^{\perp}}\] Conversely, if $\rho$ has the above structure, then the subspaces $W_{D}$ and $W_{D}^{\perp}$ are invariant under the action of $\rho$ and this implies that $\rho$ commutes with the orthogonal projection $P_{W_{D}}$. This proves $(i)$.

If $\rho=\sum_{j} \rho_{j} |u_{j}\rangle\langle u_{j}|$ for some $u_{j}\in \initsp$ commutes with $P_{W_{D}}$, then  
\[ P_{W_{D}}\rho =P_{W_{D}} \rho P_{W_{D}}= \sum_{j} \rho_{j} | P_{W_{D}} u_{j}\rangle \langle P_{W_{D}} u_{j}|\] that has the form (\ref{ecu-state}) and, hence, its is invariant. Consequently, if $\rho$ is invariant then  $P_{W_{D}^{\perp}} \rho P_{W_{D}^{\perp}}= \rho - P_{W_{D}}\rho P_{W_{D}}$ is invariant due to the linearity of the infinitesimal generator ${\mathcal L}_{*}$. Moreover, the support of $P_{W_{D}} \rho P_{W_{D}}= P_{W_{D}} \rho$ is a subset of $W_{D}$ and the support of  $ P_{W_{D}^{\perp}}\rho P_{W_{D}^{\perp}}= P_{W_{D}^{\perp}}\rho$ is orthogonal with the former. Item $(ii)$ readily follows after normalization. \\   
\end{proof}

\begin{theorem}\label{detailed-balance-states} Let ${\mathcal L}_{*}$ be a WCLT generator associated with a (possible degenerate) reference Hamiltonian $H$ and let $\rho$ be a state. If for every Bohr frequency $\omega$ the following condition holds 
\begin{eqnarray}\label{quasi-commute}
\rho D_{\omega}= c_\omega D_{\omega} \rho,
\end{eqnarray}  for some $c_{\omega}\neq 0$ that may depend on the Bohr frequencies, then,  
\begin{itemize}
\item[(i)] $\rho$ commutes with $D_{\omega}D_{\omega}^{*}$ and $D_{\omega}^{*}D_{\omega}$ for every $\omega$.
\item[(ii)] $\rho$ has the structure (\ref{structure-rho-1}) in the above theorem
\item[(iii)] if moreover $c_{\omega}=\frac{\Gamma_{-, \omega}}{\Gamma_{+,\omega}}, \; \forall \; \omega$, then $\rho$ is a stationary state.
\end{itemize}
In particular, any state $\rho$ satisfying (\ref{quasi-commute}) belongs to the annihilator of all Kraus operators.  
\end{theorem}
\begin{proof} Item $(i)$, follows immediate from (\ref{quasi-commute}). 
 
If (\ref{quasi-commute}) holds for every Bohr frequency, then $u\in \ker D_{\omega}$ implies that $0= \rho D_{\omega}u = c_{\omega} D_{\omega} \rho u$, meaning that $\rho u \in \ker D_{\omega}$. Likewise  for every $u\in \initsp$ one has that $\rho D_{\omega}^{*} u = c_{\omega}^{-1} D_{\omega}^{*} \rho u$, then $v= D_{\omega}^{*}u\in  \textrm{Im} D_{\omega}^{*}=(\ker D_{\omega})^{\perp}$ implies that $\rho v \in \textrm{Im} D_{\omega}^{*}$. This implies that the subspaces $\ker D_{\omega}$ and $(\ker D_{\omega})^{\perp}$ are $\rho$-invariant. Therefore $\rho$ commutes with each orthogonal projection $P_{\ker D_{\omega}}$, consequently, commutes with $P_{W_{D}}= \Pi_{\omega\in B_{+}} P_{\ker D_{\omega}} P_{\ker D_{\omega}^{*}}$. As a consequence of item $(i)$ in the above theorem, $\rho$ has the required structure (\ref{structure-rho-1}) and this proves $(ii)$. 

To prove $(iii)$ observe that for any state satisfying (\ref{quasi-commute}) with $c_{\omega}=\frac{\Gamma_{-, \omega}}{\Gamma_{+,\omega}}$ for all $\omega$ we have 
\begin{eqnarray}
\begin{aligned}
{\mathcal L}_{\omega *}(\rho) = & \Big( - \frac{\Gamma_{-,\omega}}{2} \{ D_{\omega}^{*}D_{\omega}, \rho \}  + \Gamma_{+, \omega} D_{\omega}^{*} \rho  D_{\omega}\Big) \\ & + \Big( - \frac{\Gamma_{+,\omega}}{2} \{D_{\omega}D_{\omega}^{*}, \rho\} + \Gamma_{-, \omega} D_{\omega} \rho  D_{\omega}^{*} \Big) - i [\Delta_{\omega}, \rho]  \\  = &
\big(-\Gamma_{-,\omega} + c_{\omega} \Gamma_{+, \omega} \big)D_{\omega}^{*}D_{\omega} \rho + \big(-\Gamma_{+,\omega} + c_{\omega}^{-1} \Gamma_{-, \omega} \big) \rho D_{\omega}^{*}D_{\omega}  \\ & - i \rho \Big( \zeta_{-, \omega} D_{\omega}^{*}D_{\omega} + \zeta_{+, \omega} D_{\omega} D_{\omega}^{*}  - c_{\omega}^{-1} c_{\omega} \zeta_{-, \omega} D_{\omega}^{*}D_{\omega} \\ & - \zeta_{+, \omega} c_{\omega} c_{\omega}^{-1} D_{\omega} D_{\omega}^{*}\Big)=0
\end{aligned}
\end{eqnarray} for every Bohr frequency. This proves that $\rho$ is invariant.

Finally, inequality $\Gamma_{-,\omega} > \Gamma_{+, \omega}$ implies that $c_{\omega} \neq 1$ for all $\omega$, then condition $\rho D_{\omega}= c_{\omega} D_{\omega} \rho$ yields $\textrm{tr}(\rho D_{\omega})= c_{\omega} \textrm{tr}(D_{\omega} \rho)= c_{\omega} \textrm{tr}(\rho D_{\omega})$, which in turn implies that $\textrm{tr}(\rho D_{\omega})=0$. Consequently $\rho$ belongs to the annihilator of all Kraus operators. This finishes the proof.
\end{proof}

\begin{definition}\label{def-db}
\begin{itemize}
\item[(i)] Any ${\mathcal L}_{*}$-invariant state satisfying (\ref{quasi-commute}) will be called \textit{local detailed balance invariant state}.  
\item[(ii)] If moreover, in (\ref{quasi-commute}) $c_{\omega}=\frac{\Gamma_{-, \omega}}{\Gamma_{+,\omega}}, \; \forall \; \omega$, the corresponding ${\mathcal L}_{*}$-invariant state will be called \textit{detailed balance invariant state}.
\end{itemize}
\end{definition} 
\begin{remark}
Any detailed balance invariant state is local detailed balance, but there exists local detailed balance states that are not detailed balance, see for instance Ref.\cite{Cruz-Q,gg} 
\end{remark}

The result of the following lemma is essentially contained in the proof of Proposition 7 in Ref.\cite{BolFF}

\begin{lemma}\label{character-WD}
$W_{D}=\bigcap_{\omega\in B_{+}}\ker \Big(D_{\omega}D_{\omega}^{*} + D_{\omega}^{*} D_{\omega}\Big)$.
\end{lemma}
\begin{proof} 
Inequality $W_{D}\subset \bigcap_{\omega\in B_{+}}\ker \Big(D_{\omega}D_{\omega}^{*} + D_{\omega}^{*} D_{\omega}\Big)$ is evident. \\  

Now, if $D_{\omega}D_{\omega}^{*} u + D_{\omega}^{*} D_{\omega}u =0$, then identity \[0=\langle  u , D_{\omega} D_{\omega}^{*}u\rangle + \langle u, D_{\omega}^{*} D_{\omega} u\rangle = \|D_{\omega}^{*}u\|^{2} + \|D_{\omega}u\|^{2}\] imply that $\|D_{\omega}^{*}u\|^2 =0$ and $\|D_{\omega}u\|^{2}=0$, hence $u\in \ker D_{\omega}\cap \ker D_{\omega}^{*}\subset W_{D}$. 
\end{proof}

\begin{proposition}
The orthogonal projection $P_{W_{D}}$ on $W_{D}$, belongs to the commutant of $H$. 
\end{proposition}
\begin{proof}
By Lemma \ref{character-WD}, we have that 
\begin{eqnarray}\label{WD-1}
W_{D}=\bigcap_{\omega\in\mathbf{B}_{+}}\ker \left(D_{\omega}D_{\omega}^{*} + D_{\omega}^{*}D_{\omega}\right)= \bigcap_{\omega\in\mathbf{B}_{+}}\ker \left(D_{\omega}D_{\omega}^{*}\right) \cap \ker\left( D_{\omega}^{*}D_{\omega}\right)
\end{eqnarray} where for every $\omega = \epsilon_{n}-\epsilon_{m}$, 
\begin{eqnarray}\label{D}
D_{\omega}= \sum_{(\epsilon_{n}, \epsilon_{m})\in B_{+,\omega}}P_{\epsilon_m}D P_{\epsilon_n}
\end{eqnarray} 

We affirm that each spectral projection $P_{\epsilon_k}$ commutes with $P_{W_{D}}$, i.e., $P_{W_{D}}\in \{H\}'$. Since \[P_{W_{D}}=\Pi_{\omega\in B_{+}} P_{\ker (D_{\omega}D_{\omega}^{*})} P_{\ker (D_{\omega}^{*}D_{\omega})}\] it suffices to prove that every spectral projection commutes with $P_{\ker (D_{\omega}D_{\omega}^{*})}$ and $P_{\ker (D_{\omega}^{*}D_{\omega})}$ for all $\omega\in B_{+}$.

Due to Lemma 3.1 in Ref.\cite{AccardiFQ} condition $(\epsilon_{n}, \epsilon_{m}), (\epsilon_{n'}, \epsilon_{m})\in B_{+, \omega}$ implies $\epsilon_{n}=\epsilon_{n'}$. Hence, each $P_{\epsilon_m}$ appears on the left of only one summand of (\ref{D}). Moreover, one can order the set $B_{+,\omega}:=\{\epsilon_{m}: \epsilon_{m} + \omega \in \textrm{Sp}(H)\}$ so that $B_{\omega}=\{\epsilon_{0}(\omega)<\epsilon_{1}(\omega)< \cdots \}$. From now on we write simply $\epsilon_{k}$ instead $\epsilon_{k}(\omega)$ and we can write $D_{\omega}= \sum_{\epsilon_{m}\in B_{+, \omega}}P_{\epsilon_m} D P_{\epsilon_m +\omega}$. Now, after simple computations we get  
\begin{eqnarray}
\begin{aligned}
D_{\omega}D_{\omega}^{*} = & \sum_{\epsilon_{m}\in B_{+, \omega}} \sum_{\epsilon_{m'}\in B_{+, \omega}} P_{\epsilon_m} D P_{\epsilon_m +\omega}P_{\epsilon_{m'} + \omega} D^{*} P_{\epsilon_{m'}} \\ =& \sum_{\epsilon_{m}\in B_{+, \omega}}  P_{\epsilon_m} D P_{\epsilon_m +\omega} D^{*} P_{\epsilon_{m}}
\end{aligned}
\end{eqnarray} Orthogonality of ranges of the above summands readily implies that 
\[P_{\ker(D_{\omega}D_{\omega}^{*})}= \Pi_{\epsilon_{m}\in B_{+, \omega}} P_{\ker (P_{\epsilon_m}D P_{\epsilon_m + \omega} D^{*} P_{\epsilon_m})}\]
Clearly $\ker P_{\epsilon_{m}} \subseteq \ker (P_{\epsilon_m}D P_{\epsilon_m + \omega} D^{*} P_{\epsilon_m})$, hence for every $m\geq 0$,
\[P_{\ker(P_{\epsilon_m}D P_{\epsilon_m + \omega} D^{*} P_{\epsilon_m})}=P_{\ker P_{\epsilon_{m}}}+ P_{\ker(P_{\epsilon_m}D P_{\epsilon_m + \omega} D^{*} P_{\epsilon_m})\cap \textrm{Im}P_{\epsilon_{m}}}\] with the second summand being a subprojection of $P_{\epsilon_{m}}$. Therefore we have that 
\begin{eqnarray}
\begin{aligned}
P_{\epsilon_l}P_{\ker (P_{\epsilon_m}D P_{\epsilon_m + \omega} D^{*} P_{\epsilon_m})}&=P_{\epsilon_l}= P_{\ker (P_{\epsilon_m}D P_{\epsilon_m + \omega} D^{*} P_{\epsilon_m})} P_{\epsilon_l}, \; \; l\neq m \\  
P_{\epsilon_m}P_{\ker(P_{\epsilon_m}D P_{\epsilon_m + \omega} D^{*} P_{\epsilon_m})}&= P_{\ker (P_{\epsilon_m}D P_{\epsilon_m + \omega} D^{*} P_{\epsilon_m})\cap \textrm{Im}P_{\epsilon_{m}}} \\ &=  P_{\ker(P_{\epsilon_m}D P_{\epsilon_m + \omega} D^{*} P_{\epsilon_m})} P_{\epsilon_m}
\end{aligned}
\end{eqnarray}
This proves the commutativity of each spectral projection with every \\ $P_{\ker (P_{\epsilon_m}D P_{\epsilon_m + \omega} D^{*} P_{\epsilon_m})}$ and, consequently, with $P_{\ker(D_{\omega}D_{\omega}^{*})}$. 

The commutativity with $P_{\ker (D_{\omega}^{*}D_{\omega})}$ is proved in a similar way.
\end{proof}

\begin{remark}
As a consequence of the above Proposition, one can add the term $i[\Delta_{1}, \rho]$ to any WCLT generator, with $\Delta_{1}$ any real function of the Hamiltonian, and if $\rho\in\{H\}'$, still have that $P_{W_{D}}\rho P_{W_{D}}$ is invariant.
\end{remark}

\begin{remark}\textbf{(Quantum detailed balance)}
In the case when a detailed balance state $\rho$ is faithful, condition (\ref{quasi-commute}) yields a privileged GKSL representation of ${\mathcal L}$, in the sense of Ref.\cite{FF-VU} Therefore, in this case, our detailed balance condition implies the quantum detailed balance condition ${\mathcal L} - \tilde{{\mathcal L}} = 2 i [\Delta_{\omega}, \cdot]$ where $\tilde{{\mathcal L}}$ is the $0$-dual generator. Moreover, in this case, 
\begin{itemize}
\item The $0$-dual ${\mathcal T}$ is a quantum Markov semigroup and 
\item The Markov generator ${\mathcal L}$ commutes with $\sigma_{-i}$, where $(\sigma_{t})$ is the associated modular group $\sigma_{t}(x)=\rho^{it}x \rho^{-it}$, see Ref.\cite{FF-VU}
\end{itemize}
\end{remark}
\begin{remark}
As we have shown in Theorem \ref{detailed-balance-states}, condition (\ref{quasi-commute}) implies that the state $\rho$ belongs to $Ann(D)$. The reciprocal is not true, indeed, taken $H=\epsilon_{0}|e_{0}\rangle \langle e_{0}| + \epsilon_{1}\Big(|e_{1}\rangle \langle e_{1}|+ |e_{2}\rangle \langle e_{2}|\Big)$ and $D=|e_{0}\rangle \langle e_{1}|+ |e_{1}\rangle \langle e_{2}|$, one can easily see that the state $\rho= \frac{1}{2} |e_{0}\rangle\langle e_{0}| +  \frac{1}{4} |e_{1}+ e^{i\theta} e_{2}\rangle \langle  e_{1}+ e^{i\theta} e_{2}|$ belongs to $Ann(D)$, but doesn't exists  a constant $c_{\omega}$ such that $\rho D_{\omega}= c_{\omega} D_{\omega}\rho$ for all $\omega$.
\end{remark}
 
\subsection{Examples}\label{exam}
\begin{example}\label{degenerate-H}
(\textbf{Degenerate $H$})
Consider a (possible degenerate) reference Hamiltonian with the property that for every Bohr frequency $\omega$ there exists a unique pair of eigenvalues $(\epsilon_n , \epsilon_m)$ such that $\omega=\epsilon_n - \epsilon_m$ (if moreover $H$ is non-degenerate then it is called generic). Let ${\mathcal L}$ be a Markov generator associated with H, then $D_{\omega}=P_{m}DP_{n}$ if $\omega=\epsilon_{n}-\epsilon_{m}$, with not necessarily rank-one spectral projections $P_{n}$. Let $\rho$ be a faithful invariant state belonging to $\{H\}'\subset \{H\}''$, since for every $\omega$, $\rho D_{\omega}=\frac{\rho_{m}}{\rho_{n}} D_{\omega}\rho$, the \textit{detailed balance condition} in the above theorem holds true iff $\frac{\rho_{m}}{\rho_{n}}=\frac{\Gamma_{-, \omega}}{\Gamma_{+, \omega}}$ for every $\omega$; consequently, the state is invariant and has the (block) structure given in Theorem \ref{detailed-balance-states}. Moreover, this condition yields the (block) Gibbs state \[\rho= \frac{1}{Z} e^{-\beta(H)H}\] 
\end{example} 

\begin{example}\label{photosynthesis}
(\textbf{Dark states in quantum photosynthesis (KV's model)}) The invariant states of a WCLT generator modeling quantum photosynthesis, proposed by Kozyrev and Volovich\cite{Kozyrev-Vol} in the context of stochastic limit approach of degenerate quantum open systems, were studied in Ref.\cite{gggq} We will show that the so called \textit{dark stationary states} of that model correspond with states supported on the interaction-free subspace.  

In this case the reference Hamiltonian has the spectral decomposition \[H= \epsilon_{0}P_{0} + \epsilon_{1} P_{1} + \epsilon_{2} P_{2}\] where $P_{0}=|e_{0}\rangle \langle e_{0}|, \; P_{1}=|e_{1}\rangle \langle e_{1}|$, $P_{2}= \sum_{j=2}^{N} |e_{j}\rangle \langle e_{j}|$ and $\epsilon_{2} > \epsilon_{1} > \epsilon_{0}=0$. The set of Bohr frequencies is $B_{+}=\{\omega_{1}=\epsilon_{2}, \omega_{2}=\epsilon_{2}-\epsilon_{1}, \omega_{3}=\epsilon_{1}\}$. At level two there are two linearly dependent maximal entangled vectors: the bright photonic vector $\chi= \sum_{j=2}^{N} e_{j}$ and the bright phononic vector $\Psi= e^{i \theta}\chi$, $\theta\in (0,2\pi)\setminus \{\pi\}$, so that the corresponding pure states coincide $|\chi\rangle \langle \chi| = |\Psi\rangle \langle \Psi|$.  

Up to some constants the interaction operators are:
\begin{eqnarray}
\begin{aligned}
D_{\omega_{1}}= & | e_0\rangle \langle \chi | \\ 
D_{\omega_{2}}= & | e_1\rangle \langle \Psi | \\ 
D_{\omega_{3}} = & | e_0\rangle \langle e_{1}|
\end{aligned}
\end{eqnarray} Therefore $W_{D}=\{e_{0}, e_{1}, \chi\}^{\perp}$. Consequently, any invariant state supported on the interaction-free subspace $W_{D}$ is orthogonal with the bright photonic state $|\chi\rangle\langle \chi|$ and also orthogonal with the pure states $|e_{0}\rangle \langle e_{0}|$ and $|e_{1}\rangle \langle e_{1}|$, hence, it is a \textit{global} dark state in the sense of Ref.\cite{Kozyrev-Vol}, i.e., dark with respect to the full generator ${\mathcal L}$. This result coincides with the result found in Ref.\cite{gggq}, where it was proved that any invariant state has the spectral decomposition 
\begin{eqnarray}\label{mixed-dark}
\rho= \lambda\Big(r_{0} P_{0} + r_{1} P_{1} + \frac{1}{N-1} |\chi\rangle \langle \chi|\Big) + \big(1-\lambda (r_{0}+r_{1}-1)\big) \rho_{W_D}
\end{eqnarray} 
where $\rho_{W_D}$ is a dark state, $r_{0}$, $r_{1}$ are functions of the $\Gamma's$ and $0\leq \lambda \leq (1+ r_{0} + r_{1})^{-1}$. Indeed, assuming that $\rho$ is invariant and defining \[\big(1-\lambda (r_{0}+r_{1}-1)\big) \rho_{W_D} := \rho - \lambda\Big(r_{0} P_{0} + r_{1} P_{1} + \frac{1}{N-1} |\chi\rangle \langle \chi|\Big)\] one has that $\rho_{W_D}$ is:  
\begin{itemize} 
\item supported on (just complete the diagonalization of $\rho$) $W_{D}=\{e_{0}, e_{1}, \chi\}^{\perp}$ consequently, it is invariant, and 
\item orthogonal w.r.t. the Hilbert-Schmidt product with $P_{0}$, $P_{1}$ and $|\chi\rangle \langle \chi|$, then it is global dark. 
\end{itemize}
\end{example}

\section{Invariant states of a modified AKV's model}\label{AKV-model} 

The invariant states of a slightly modification of a quantum transport model due to Aref'eva, Kozyrev and Volovich\cite{Arefeva}, were characterized in Ref.\cite{gggq} in terms of some conditions including commutation relations with remarkable operators of the model. 

In this section we will give a full description of the set of invariant states of the model, including their parametrization in terms of states supported on the range of $|Z|^{2}=Z^{*}Z$, where $Z$ is the so called interference operator. It turns out that $Z$ plays  the role of a ``generalized annihilation operator'' mapping the second eigen-space of the reference Hamiltonian $H$ into the third one, this analogy allows us to compute all invariant states by a simple method and notably less computations that those performed in the previous reference. Moreover we identify the fast recurrent subspace and study the approach to equilibrium property.

\subsection{The modified AKV's quantum transport model}
The corresponding generator belongs to the class of WCLT and the associated reference Hamiltonian is given by 
 \[H=0 P_{0} + \epsilon_{1} P_{1} + \epsilon_{2} P_{2} + \epsilon_{3} P_{3}\] with $P_{0}=|e_{0}\rangle \langle e_{0}|$, \, $P_{1}=|e_{1}\rangle \langle e_{1}|$, \, $P_{2}= \sum_{j=2}^{N}|e_{j}\rangle \langle e_{j}|$, \, $P_{3}=\sum_{j=N+1}^{N+M} |e_{j}\rangle \langle e_{j}|$, $M\leq N-1$ and the eigenvalues satisfy the inequalities $0 < \epsilon_{3} < \epsilon_{2} < \epsilon_{1}$. The interaction operator is  
\[D= |Te_{N+1}\rangle \langle Te_{1}| + Z + |T e_{0}\rangle \langle T e_{N-1}| \] where $Z$ is an interference operator performing transitions between levels $\epsilon_{2}$ and $\epsilon_{3}$ \[Z:P_{2}(\initsp) \mapsto P_{3}(\initsp)\] defined in terms of some interference coefficients $g_{ab}$ (see Remark 2.2 in Ref.\cite{gggq}) 
as 
\[Z= \frac{1}{\sqrt{N-1}}\sum_{b=N+1}^{N+M} \sum_{a=2}^{N} g_{ab} |b\rangle \langle a|\] and  $T=P_{0} + P_{1} + Z + Z^{*}$. 
The maximally entangled states are defined, up to normalization, as $|\psi\rangle\langle \psi |$ with $\psi=\sum_{a=2}^{N} |a\rangle$ and $|\psi'\rangle\langle \psi'|$ with $\psi'= \sum_{b=N+1}^{N+M} |b\rangle$. Writing simply $T_{a}$ instead $T|a\rangle$, one can observe that $\sqrt{N-1} Z_{N-1}=\psi'$ and $\sqrt{N-1}Z^{*}_{N+1}=\psi$. 

Hence, there exist six Bohr frequencies: $\omega_{1}=\epsilon_{1}-\epsilon_{2}, \; \omega_{2}=\epsilon_{2}-\epsilon_{3}, \; \omega_{3}=\epsilon_{3}, \; \omega_{4}=\epsilon_{1}, \; \omega_{5}=\epsilon_{2}, \; \omega_{6}=\epsilon_{1}-\epsilon_{3}$. But only three frequencies have associated non-trivial generators, so that  
\[{\mathcal L}= {\mathcal L}_{\omega_{1}} + {\mathcal L}_{\omega_{2}} + {\mathcal L}_{\omega_{3}}\] \, and \, ${\mathcal L}_{\omega_{4}}= {\mathcal L}_{\omega_{5}} = {\mathcal L}_{\omega_{6}}=0$.

\begin{remark}\label{Z}
The relations $T=T^{*}$, $Z^{2}=0$ and $Z^{*2}=0$ immediately follow. Moreover, by direct computation one can show that $|Z|$ is an orthogonal projection, $TP_{2}T=P_{3}=ZZ^{*}$ and $TP_{3}T=|Z|$, see Corollary 2.1 in Ref.\cite{gggq}
\end{remark}
One can see that 
\begin{eqnarray}
\begin{aligned}
D_{\omega_{1}}&= \frac{1}{\sqrt{N-1}}|\psi\rangle \langle e_{1}| \\ D_{\omega_{2}}&= Z \\  
D_{\omega_{3}}&=  \frac{1}{\sqrt{N-1}}|e_{0}\rangle \langle \psi'|
\end{aligned}
\end{eqnarray}
So that the Kraus operators are 
\begin{eqnarray}
\begin{aligned}
L^{\omega_{1}}_1&=\sqrt{2\Gamma_{Re, -, \omega_{1}}} \:  |\psi\rangle \langle e_{1}| \\�L^{\omega_{2}}_1&=\sqrt{2(N-1)\Gamma_{Re, -, \omega_{2}}} \: Z \\  
L^{\omega_{3}}_1&= \sqrt{2\Gamma_{Re, -, \omega_{3}}} \: |e_{0}\rangle \langle \psi'| \\ 
L^{\omega_{1}}_2&=\sqrt{2\Gamma_{Re, +, \omega_{1}}} \:  |e_{1}\rangle \langle \psi| \\ L^{\omega_{2}}_2&=\sqrt{2(N-1)\Gamma_{Re, +, \omega_{2}}} \: Z^{*} \\  
L^{\omega_{3}}_2&= 0
\end{aligned}
\end{eqnarray} with all coefficients $\Gamma_{Re,\pm,\omega_{l}}, \; \; 1\leq l\leq 2$ positive, $\Gamma_{Re, -,\omega_{3}}>0$ and \\ $\Gamma_{Re,+,\omega_{3}}=0$. \\ 

The effective Hamiltonian is give by 
\begin{eqnarray}\label{effect-H}
\begin{aligned}
H_{eff}= &\Gamma_{Im,+,\omega_{1}}  
|\psi\rangle\langle \psi| -(N-1)\Gamma_{Im,-,\omega_{1}} |e_{1}\rangle\langle e_{1}|  \\ + &(N-1)\Gamma_{Im, +, \omega_{2}} P_{3} -(N-1)\Gamma_{Im,-,\omega_{2}}|Z| \\ - & \Gamma_{Im,-,\omega_{3}}|\psi' \rangle \langle \psi'|
\end{aligned}
\end{eqnarray} with all coefficients $\Gamma_{Im,\pm,\omega_{l}}, \; \; 1\leq l\leq 3$, positive.

Direct computation shows that  
\begin{eqnarray}\label{WD}
\begin{aligned}
W_{D} & = \{e_{1}, \psi, \psi'\}^{\perp}\cap \ker Z \cap \ker Z^{*} \\ & =\{e_{1}, \psi, \psi'\}^{\perp}\cap\{Z^{*} e_{j}: j= N+1, \cdots N+M\}^{\perp}\cap\{Z e_{j}: j= 2, \cdots N\}^{\perp}
\end{aligned}
\end{eqnarray}

The vectors $\psi$ and $Z^{*}\psi'$ are fixed points of the orthogonal projection $|Z|$. Indeed, 
\[|Z|\psi = Z^{*}Z \sqrt{N-1}Z^{*} e_{N+1}= \sqrt{N-1}Z^{*}e_{N+1}=\psi\] and 
\[|Z| Z^{*}\psi' = Z^{*} Z Z^{*}\psi' = Z^{*}\psi'\] where we have used the identity $ZZ^{*}=P_{3}$. Hence we can write 
\[\textrm{Im}|Z| = \textrm{span}\{\psi, Z^{*}\psi'\}\oplus\textrm{Im}|Z|\cap \{\psi, Z^{*}\psi'\}^{\perp}\] Moreover, identity  
$P_{2}(\initsp)= \textrm{Im}|Z| \oplus \ker|Z|$ yields the decomposition \[P_{2}(\initsp)=\textrm{span}\{\psi, Z^{*}\psi'\}\oplus\Big(\textrm{Im}|Z|\cap \{\psi, Z^{*}\psi'\}^{\perp}\Big)\oplus \ker|Z| \]

Similarly one can see that $\psi'$, $Z \psi$ are fixed points of $P_{3}$ and the decomposition 
\[P_{3}(\initsp) = \textrm{span}\{\psi', Z\psi\}\oplus \textrm{Im}P_{3}\cap\{\psi', Z\psi\}^{\perp}\] holds true.

\subsection{Subharmonic projections and structure of sta\-tio\-na\-ry states}
A positive operator $a$ is subharmonic (resp. superharmonic, resp. harmonic) for a quantum Markov semigroup $({\mathcal T})_{t\geq 0}$ if ${\mathcal T}_{t}(a)\geq a$ (resp. ${\mathcal T}_{t}(a)\leq a$, resp. ${\mathcal T}_{t}(a)= a$) for all $t\geq 0$.  Subharmonic projections play a fundamental role in the study of irreducibility of the semigroup, see Ref.\cite{FF-RR}, in particular, orthogonal projections onto the support subspaces of invariant states are subharmonic. We will compute harmonic and subharmonic projections of the modified AKV's model as a first step in the characterization of invariant states.

\begin{proposition}
Each rank-one projection $|w\rangle \langle w|$ with $w\in W_{D}$ is harmonic for the AKV's semigroup, consequently $P_{W_{D}}$ and $P_{W_{D}^{\perp}}$ are by itself harmonic projections. 
\end{proposition} 
\begin{proof}
It suffices to prove that every rank-one projection $|w\rangle \langle w|$ with $w\in W_{D}$ commutes with the effective Hamiltonian $H_{eff} = H_{\omega_{1}} + H_{\omega_{2}} + H_{\omega_{3}}$. But clearly $|w\rangle \langle w|$ commutes with $H_{\omega_{1}}$ since $w\in \{e_{1}, \psi\}^{\perp}$. Similarly it commutes with $H_{\omega_{3}}$ since $w \perp \psi'$. Now, for every $v\in P_{3}(\initsp)$ identity $Z^{*}v=0$ implies that $v=P_{3}v = Z Z^{*} v=0$, hence $\ker Z^{*}= \ker P_{3}$, it readily follows that $P_{3}w=0$. Moreover, $\ker |Z| = \ker Z$, hence $w\in \ker Z$. This implies that $|w\rangle \langle w|$ commute with $H_{\omega_{2}}$ and, hence, commute with $H_{eff}$. We can conclude that every $|w\rangle \langle w|$ is harmonic. Being a linear combination of the above class of rank-one projections, $P_{W_{D}}$ is harmonic projection by itself. Moreover, due to the conservativity, we have that \[{\mathcal T}_{t}(P_{W_{D}^{\perp}})= {\mathcal T}_{t}(I) -{\mathcal T}_{t}(P_{W_{D}})= I- P_{W_{D}} = P_{W_{D}^{\perp}}\] This finishes the proof. 
\end{proof}

But apart from the above given projections, the AKV's semigroup has another subharmonic projections, as we show in the following. 

\begin{proposition} The orthogonal projection $p_V$ onto the subspace \\ $V=\{e_{0}, e_{1}, \psi, \psi', Z\psi, Z^{*}\psi' \}^{\perp}$ is  subharmonic for the quantum Markov semigroup of the modified AKV's model. 
\end{proposition}
\begin{proof}
By Theorem III.1 in Ref.\cite{FF-RR}, it suffices to prove that the orthogonal projection $p_{V}$  satisfies that $Ran(p_{V})$ is invariant under the action of the semigroup generated by the effective Hamiltonian and satisfy 
\begin{eqnarray}\label{cons-suharm}
L_{k}^{\omega_{j}}u = p_{V} L_{k}^{\omega_{j}}u, \; \; 0\leq j\leq 3, \; \; 1\leq k\leq 2, \; \; \forall \; \; u\in Ran(p_{V})
\end{eqnarray} 
These conditions holds true if and only if $V=Ran(p_{V})$ is invariant under the action of the semigroup generated by the effective Hamiltonian $H_{eff}=H_{\omega_{1}}+ H_{\omega_{2}} + H_{\omega_{3}}$ and the following relations hold for all $u\in V$: 
\begin{itemize}
\item[(i)] $\langle e_{1}, u\rangle \psi = \langle e_{1}, u\rangle p_{V}\psi$, 
\item[(ii)] $\langle \psi, u\rangle e_{1} = \langle \psi, u\rangle p_{V} e_{1}$
\item[(iii)] $p_{V} Zu = Zu$, 
\item[(iv)] $p_{V} Z^{*}u = Z^{*}u$
\item[(v)] $\langle \psi', u\rangle e_{0} = \langle \psi', u\rangle p_{V} e_{0}$.
\end{itemize} Due to the explicit form of  $H_{eff}=H_{\omega_{1}} + H_{\omega_{2}}+H_{\omega_{3}}$, it turns out that conditions $(i)-(v)$ imply that $V$ is invariant under the action of the unitary group generated by this operator. Indeed, the only nontrivial fact is to show that $p_{V}$ commutes with 
\[H_{\omega_{2}}=(N-1)\Gamma_{Im, +, \omega_{2}} TP_{2}T-(N-1)\Gamma_{Im,-,\omega_{2}}TP_{3}T\]
But $TP_{2}T=P_{3}$ and $TP_{3}T=|Z|$ and for $u\in V$ one has $P_{3}u=ZZ^{*}u=Zp_{V}Z^{*}u=p_{V}Zp_{V}Z^{*}u=p_{V}ZZ^{*}u=p_{V}P_{3}u$, hence $P_{3}p_{V}=p_{V}P_{3}$. Similarly, for $u\in V$, using that $|Z|$ is a projection one has $|Z|u=Z^{*}Zu=Z^{*}p_{V}Zu=p_{V}Z^{*}p_{V}Zu=p_{V}|Z|u$, i.e., $|Z|p_{V}=p_{V}|Z|$.  

Now, if $p_{V}$ is the orthogonal projection on the subspace $V$ conditions $(i), (ii), (v)$ holds true for e\-ve\-ry $u\in V$. To show that conditions $(iii)$ and $(iv)$ hold, it suffices to prove that  $Zu, Z^{*}u\in V$ for e\-ve\-ry $u\in V$. But for every $u\in V$, it readily follows that $Zu, Z^{*}u \in \{e_{0}, e_{1}, \psi, \psi'\}^{\perp}$. Moreover, notice that $\langle Zu, Z\psi\rangle=\langle u, |Z|^{2}\psi\rangle =\langle u, \psi\rangle=0$ since $|Z|^{2} \psi = \sqrt{N-1} Z^{*}Z Z^{*}e_{N+1} = \sqrt{N-1} Z^{*}P_{3} e_{N+1}= \sqrt{N-1} Z^{*}e_{N+1}=\psi$ (i.e., $|Z|\psi =\psi$), and $\langle Zu, Z^{*}\psi'\rangle=\langle Z^{2}u, \psi'\rangle =0$, since $Z^{2}=0$, hence $Zu\in V$. Finally, $Z^{2}=0$ implies that $\langle Z^{*}u, Z\psi\rangle = \langle u, Z^{2}\psi\rangle=0$ and $\langle Z^{*}u, Z^{*}\psi'\rangle = \langle u, ZZ^{*}\psi'\rangle =\langle u, P_{3}\psi'\rangle = \langle u, \psi'\rangle=0$, this shows that $Z^{*}u\in V$ and finishes the proof.
\end{proof}

\subsection{Invariant states}
Now we look for invariant states associated with the non-trivial subharmonic projections. As a consequence of our Theorem \ref{theorem-general-hamiltonian-11}, every state supported on a subspace of $W_{D}$ is necessarily invariant. The existence of invariant states with support contained in $W_{D}^{\perp}$ is a question that naturally arises. To answer this question we start with the following.

\begin{lemma}\label{lemma-V} 
A state $\rho$ supported on a subspace of $V\cap{W_{D}^{\perp}}$ is invariant if and only if $\{e_{1},\psi, \psi'\} \in \ker \rho$, $\rho$ belongs to the commutant $\{H\}'$ of the reference Hamiltonian, commute with $|Z|$ and satisfies  
\begin{eqnarray}\label{eq-zeta-rho}
Z^{*}\rho Z= \frac{\Gamma_{Re,-,\omega_{2}}}{\Gamma_{Re,+,\omega_{2}}} \rho|Z|
\end{eqnarray}
\end{lemma}
\begin{proof} 
Let $\rho$ be a state supported on a subspace of $V\cap W_{D}^{\perp}$, hence $\{e_{1}, \psi, \psi'\}\subset\ker \rho$. It readily follows that $\rho\in\ker{{\mathcal L}_{*\omega_1}}\cap \ker{{\mathcal L}_{*\omega_3}}$.  
Therefore, if $\rho$ is invariant then 
\begin{eqnarray}\label{inv-omega-two}
\begin{aligned}
0={\mathcal L}_{\omega_{2} *}(\rho)  = & -i\Big(\Gamma_{Im, +, \omega_{2}} P_{3}\rho - \Gamma_{Im,-,\omega_{2}}|Z|\rho\Big)  \\ & + i \Big(\Gamma_{Im, +, \omega_{2}} \rho P_{3} -\Gamma_{Im, -, \omega_{2}} \rho |Z|\Big) \\ 
& - \Gamma_{Re, -, \omega_{2}}\Big( |Z|\rho - 2 Z\rho Z^{*} + \rho|Z| \Big) \\ 
& -  \Gamma_{Re, +, \omega_{2}} \Big( P_{3} \rho - 2 Z^{*} \rho Z + \rho P_{3} \Big)
\end{aligned}
\end{eqnarray} Therefore after direct computations we get 
\begin{eqnarray}
\begin{aligned}
0=P_{2} {\mathcal L}_{\omega_{2} *}(\rho) P_{3} = & \Big(-\Gamma_{Re, -, \omega_{2}} + i \Gamma_{Im, -, \omega_{2}} - \Gamma_{Re, +, \omega_{2}} + i \Gamma_{Im, +, \omega_{2}}\Big) |Z| \rho P_{3} \\ &+ \Big( - \Gamma_{Re, +, \omega_{2}} + i \Gamma_{Im, +, \omega_{2}}\Big) (P_{2}-|Z|)\rho P_{3}
\end{aligned}
\end{eqnarray} Due to the positivity of the gammas and orthogonality of ranges we get that $|Z|\rho P_{3}=0$ and $(P_{2}-|Z|)\rho P_{3}=0$. It follows that $P_{2}\rho P_{3}=0$ and taking adjoints we get that also $P_{3}\rho P_{2}=0$. But  using that $\rho$ is supported on a subspace of $W_{D}^{\perp}\cap V \subset\big(P_{2} + P_{3}\big)(\initsp)$ we get, 
\begin{eqnarray}
[\rho, P_{3}] = (P_{2} + P_{3})[\rho, P_{3}] (P_{2}+ P_{3}) = P_{2}\rho P_{3} - P_{3}\rho P_{2} =0
\end{eqnarray} Hence $\rho$ commutes with $P_{3}$. 

On the other hand, using the commutativity of $\rho$ with $P_{3}$ from (\ref{inv-omega-two}) and direct computations we get 
\begin{eqnarray}
\begin{aligned}
0=P_{3} {\mathcal L}_{\omega_{2} *}(\rho) P_{3} = \Gamma_{Re, -, \omega_{2}} Z\rho Z^{*}- \Gamma_{Re, +, \omega_{2}} P_{3} \rho
\end{aligned}
\end{eqnarray} It follows that 
\begin{eqnarray}\label{zeta-rho-zeta-star}
Z\rho Z^{*} = \frac{\Gamma_{Re, +, \omega_{2}}}{\Gamma_{Re, -, \omega_{2}}} P_{3} \rho
\end{eqnarray}
Multiplication in (\ref{zeta-rho-zeta-star}) by $Z^{*}$ on the left and by $Z$ on the right yields 
\begin{eqnarray}\label{zeta-star-rho-zeta}
Z^{*}\rho Z = \frac{\Gamma_{Re, -, \omega_{2}}}{\Gamma_{Re, +, \omega_{2}}} |Z| \rho|Z|
\end{eqnarray} To prove (\ref{eq-zeta-rho}) it remains to prove that $\rho$ commutes with $|Z|$.

Now, once again from (\ref{inv-omega-two}) we get using (\ref{zeta-star-rho-zeta}) 
\begin{eqnarray}
\begin{aligned}
0=&P_{2}{\mathcal L}_{\omega_{2}, *}(\rho) P_{2} \\ = & i\Gamma_{Im, -, \omega_{2}} \big(|Z|\rho P_{2} - P_{2}\rho |Z|\big) - \Gamma_{Re, -, \omega_{2}}\big(|Z|\rho P_{2} +P_{2}\rho |Z|\big) + 2 \Gamma_{Re, +, \omega_{2}} Z^{*} \rho Z \nonumber \\ =& i\Gamma_{Im, -, \omega_{2}} \big(|Z| \rho P_{2} - P_{2}\rho |Z|\big) - \Gamma_{Re, -, \omega_{2}}\big(|Z|\rho P_{2} + P_{2}\rho |Z|\big) + 2 \Gamma_{Re, -, \omega_{2}} |Z|\rho |Z| \nonumber \\ =&
- \Gamma_{Re,-, \omega_{2}} |Z| \rho (P_{2}- |Z|) - \Gamma_{Re, -, \omega_{2}}(P_{2}- |Z|)\rho |Z| \\ & + 
i \Gamma_{Im, -, \omega_{2}} |Z| \rho (P_{2}- |Z|) - i \Gamma_{Im, -, \omega_{2}} (P_{2}- |Z|) \rho |Z|
\end{aligned}
\end{eqnarray} Due to orthogonality of ranges it follows that $|Z| \rho (P_{2}- |Z|)=0$ and \\ $(P_{2}- |Z|) \rho |Z|=0$  consequently, $|Z| \rho P_{2}= |Z|\rho|Z| = P_{2}\rho |Z|$. But, after direct computations we get 
\begin{eqnarray}
[\rho, P_{2}-|Z|]=(P_{2}+P_{3})[\rho, P_{2}-|Z|](P_{2} + P_{3})= |Z|\rho P_{2} - P_{2} \rho |Z|=0.
\end{eqnarray} This proves that $\rho$ commutes with $P_{2}-|Z|$. Similarly, using that $\rho$ commutes with $P_{3}$, we get 
\begin{eqnarray}
[\rho, |Z|]=(P_{2}+P_{3})[\rho,|Z|](P_{2} + P_{3})= P_{2} \rho |Z| -  |Z| \rho P_{2} =0.
\end{eqnarray} Therefore $\rho$ commutes with $P_{2}$ and with $|Z|$. The commutation of $\rho$ with $P_{0}$ and $P_{1}$ is immediate. Hence we can conclude that $\rho\in\{H\}'$. \\ 
Conversely, if a state $\rho$ is supported on a subspace of $V\cap W_{D}^{\perp}$, belongs to $\{H\}'$, commutes with $|Z|$ and satisfies (\ref{zeta-star-rho-zeta}), then $\rho\in\ker{\mathcal L}_{\omega_{1} *}\cap \ker{\mathcal L}_{\omega_{2} *} \cap \ker{\mathcal L}_{\omega_{3} *}$. This finishes the proof.
\end{proof}

Now we are ready to characterize the set of invariant states of the model. 
 
\begin{theorem}\label{explicit-inv-states} 
A state $\tau$ is invariant for the modified AKV's model if and only if it has the structure 
\begin{eqnarray}\label{rho-supported-V}
\begin{aligned}
\tau = \lambda \rho_{W_{D}} + (1-\lambda ) \Big( \frac{\Gamma_{\textrm{Re}, +, \omega_{2}}}{\Gamma_{\textrm{Re}, +, \omega_{2}} + \Gamma_{\textrm{Re}, -, \omega_{2}}} \sigma +   \frac{\Gamma_{\textrm{Re}, -, \omega_{2}}}{\Gamma_{\textrm{Re}, +, \omega_{2}} + \Gamma_{\textrm{Re}, -, \omega_{2}}} \; Z\sigma Z^{*} \Big)
\end{aligned}
\end{eqnarray} where $\rho_{W_{D}}$ is an invariant state supported on $W_{D}$, $\sigma: \textrm{Im}|Z| \mapsto \textrm{Im}|Z|$ is a state supported on a subspace of $V\cap \textrm{Im}|Z|$ and $0\leq \lambda \leq 1$. 
\end{theorem}
\begin{proof} 
It suffices to show that a state $\rho$ supported on a subspace of $V\cap W_{D}^{\perp}$ is invariant if and only if has the structure (\ref{rho-supported-V}) with $\lambda=0$. But, as a consequence of Lemma \ref{lemma-V}, any invariant state $\rho$ supported on $V\cap W_{D}^{\perp}$ belongs to $\{H, |Z|\}'$ and can be written in the form 
\begin{eqnarray}\label{explicit-rho}
\begin{aligned}
\rho= & (|Z| + P_{3})\rho (|Z|+ P_{3}) = |Z|\rho |Z| + P_{3}\rho P_{3} \\ = & \frac{\Gamma_{\textrm{Re}, +, \omega_{2}}}{\Gamma_{\textrm{Re}, +, \omega_{2}} + \Gamma_{\textrm{Re}, -, \omega_{2}}} \sigma +   \frac{\Gamma_{\textrm{Re}, -, \omega_{2}}}{\Gamma_{\textrm{Re}, +, \omega_{2}} + \Gamma_{\textrm{Re}, -, \omega_{2}}} \; Z\sigma Z^{*} 
\end{aligned}
\end{eqnarray} where $\sigma= \frac{\Gamma_{\textrm{Re}, +, \omega_{2}} + \Gamma_{\textrm{Re}, -, \omega_{2}}}{\Gamma_{\textrm{Re}, +, \omega_{2}}}  |Z|\rho |Z|$ and we have used (\ref{zeta-rho-zeta-star}). Clearly, up to normalization, $\sigma: \textrm{Im}|Z| \mapsto \textrm{Im}|Z|$ is a state supported on a subspace of $V\cap \textrm{Im}|Z|$. 

Conversely, if $\rho$ is of the form (\ref{rho-supported-V}) with $\lambda=0$, then $e_{1}\in \ker \rho$ and $\rho \psi = \frac{\Gamma_{\textrm{Re}, +, \omega_{2}}}{\Gamma_{\textrm{Re}, +, \omega_{2}} + \Gamma_{\textrm{Re}, -, \omega_{2}}} \sigma \psi$, then for any $u\in\initsp$ we have that $\langle u, \sigma \psi \rangle = \langle \sigma u, \psi\rangle = 0$ since $\sigma u\in V\cap\textrm{Im}|Z|$, hence $\sigma \psi=0$, proving that $\psi\in \ker\rho$. Similarly one can see that $\psi'\in \ker \rho$. It readily follows that 
$\rho\in \ker{\mathcal L}_{\omega_{1}*} \cap \ker {\mathcal L}_{\omega_{3}*}$. Now by direct computations one can see that $\rho$ commutes with $H$ and $|Z|$, indeed, since $\sigma$ commute with $P_{2}$, then 
\begin{eqnarray}
\begin{aligned}
P_{2}\rho= & \frac{\Gamma_{\textrm{Re}, +, \omega_{2}}}{\Gamma_{\textrm{Re}, +, \omega_{2}} + \Gamma_{\textrm{Re}, -, \omega_{2}}} P_{2} \sigma = \rho P_{2}, \: \: \;  \textrm{and} \\ P_{3}\rho =  & \frac{\Gamma_{\textrm{Re}, -, \omega_{2}}}{\Gamma_{\textrm{Re}, +, \omega_{2}} + \Gamma_{\textrm{Re}, -, \omega_{2}}} \; Z\sigma Z^{*}= \rho P_{3}
\end{aligned}
\end{eqnarray} Similarly,  \[|Z| \rho = \frac{\Gamma_{\textrm{Re}, +, \omega_{2}}}{ \Gamma_{\textrm{Re}, +, \omega_{2}} + \Gamma_{\textrm{Re}, -, \omega_{2}}} \sigma = \rho |Z|\] Therefore, Lemma \ref{lemma-V} implies that $\rho$ is invariant. 
This finishes the proof.  
\end{proof}

\begin{remark} Due to the above theorem, if $\tau$ is an invariant state for the modified AKV's model, then \[\tau = \lambda \tau_{W_{D}} + (1-\lambda) \tau_{W_{D}^{\perp}}, \; \, 0\leq \lambda \leq 1\] with 
\[ \tau_{W_{D}^{\perp}}=\Big(\frac{\Gamma_{\textrm{Re}, +, \omega_{2}}}{\Gamma_{\textrm{Re}, +, \omega_{2}} + \Gamma_{\textrm{Re}, -, \omega_{2}}} \sigma +   \frac{\Gamma_{\textrm{Re}, -, \omega_{2}}}{\Gamma_{\textrm{Re}, +, \omega_{2}} + \Gamma_{\textrm{Re}, -, \omega_{2}}} \; Z\sigma Z^{*} \Big) \]
 \end{remark}

The following result describes the extremal points of the subset of invariant states.

\begin{theorem}\label{extremal-states}
An invariant state $\rho$ is extremal for the modified AKV's model if and only if it has one of the following structures: 
\begin{itemize} 
\item[(i)] $\rho= |u\rangle \langle u|$ with $u\in W_{D}=\{e_{1}, \psi, \psi'\}^{\perp}\cap \ker Z \cap \ker Z^{*}$
\item[(ii)] 
\begin{eqnarray}\label{extremals-V}
\begin{aligned}
\rho = \frac{\Gamma_{\textrm{Re}, +, \omega_{2}}}{\Gamma_{\textrm{Re}, +, \omega_{2}}+\Gamma_{\textrm{Re}, -, \omega_{2}}}|u\rangle \langle u| + \frac{\Gamma_{\textrm{Re}, -, \omega_{2}}}{\Gamma_{\textrm{Re}, +, \omega_{2}}+\Gamma_{\textrm{Re}, -, \omega_{2}}} Z|u\rangle \langle u| Z^{*}
\end{aligned}
\end{eqnarray} 
\end{itemize} with $u\in V\cap\textrm{Im}|Z|$.  
\end{theorem}
\begin{proof}
Due to Theorem \ref{theorem-general-hamiltonian-11}, any state of the form $|u\rangle\langle u|$ with $u\in W_{D}$ is invariant and extremal. Conversely, if $\rho$ is an extremal invariant state supported on $W_{D}$, then it is necessarily a pure state, i.e., $\rho=|u\rangle\langle u|$ for some $u\in W_{D}$.

Clearly any state $\rho$ given by (\ref{extremals-V}), is supported on $V\cap W_{D}^{\perp}$. In view of the above Theorem \ref{explicit-inv-states}, $\rho$ is invariant since it has the structure (\ref{rho-supported-V}) with $\lambda=0$ and $\sigma=|u\rangle \langle u|$.

Now, if 
\begin{eqnarray}\label{rho-extremal}
\rho= \lambda \rho_{1} + (1-\lambda) \rho_{2}
\end{eqnarray} with $\rho_{i}, \; i=1, 2$ invariant states supported on $V\cap W_{D}^{\perp}$, let us say 
\begin{eqnarray}
\begin{aligned}
\rho_{i} = \frac{\Gamma_{\textrm{Re}, +, \omega_{2}}}{\Gamma_{\textrm{Re}, +, \omega_{2}}+\Gamma_{\textrm{Re}, -, \omega_{2}}} \sigma_{i} + \frac{\Gamma_{\textrm{Re}, -, \omega_{2}}}{\Gamma_{\textrm{Re}, +, \omega_{2}}+\Gamma_{\textrm{Re}, -, \omega_{2}}} Z \sigma_{i} Z^{*}, \, \, i=1,2 \nonumber
\end{aligned}
\end{eqnarray} with $\sigma_{i}$ supported on $V\cap \textrm{Im}|Z|$, $i=1, 2$.

After multiplication by $|Z|$ in (\ref{rho-extremal}) we get, 
\[|u\rangle\langle u| = \lambda \sigma_{1} + (1-\lambda) \sigma_{2} \] Therefore, either $\lambda=1$ and $\sigma_{1}=|u\rangle\langle u|$  or $\lambda=0$ and $\sigma_{2}=|u\rangle\langle u|$, since $|u\rangle\langle u|$ is a pure state. This proves that $\rho$ is extremal. Conversely, if $\rho$ is an extremal invariant state supported on $V\cap W_{D}^{\perp}$, then it has the structure (\ref{rho-supported-V}) with $\lambda=0$. Hence, necessarily   
\begin{eqnarray}
\begin{aligned}
\rho = \frac{\Gamma_{\textrm{Re}, +, \omega_{2}}}{\Gamma_{\textrm{Re}, +, \omega_{2}}+\Gamma_{\textrm{Re}, -, \omega_{2}}}\sigma + \frac{\Gamma_{\textrm{Re}, -, \omega_{2}}}{\Gamma_{\textrm{Re}, +, \omega_{2}}+\Gamma_{\textrm{Re}, -, \omega_{2}}} Z \sigma Z^{*}
\end{aligned}
\end{eqnarray} with $\sigma$ an state supported on $\textrm{Im}|Z| \cap V$, let us say $\sigma=\sum_{j} \sigma_{j}|u_{j}\rangle \langle u_{j}|$ with $(u_{j})_{j}\subset \textrm{Im}|Z|\cap V$ the orthonormal basis of $\sigma$ and $\sum_{j}\sigma_{j}=1$. Consequently, 
\begin{eqnarray}\label{sum-rho}
\rho = \sum_{j}\sigma_{j} \Big(\frac{\Gamma_{\textrm{Re}, +, \omega_{2}}}{\Gamma_{\textrm{Re}, +, \omega_{2}}+\Gamma_{\textrm{Re}, -, \omega_{2}}} |u_{j}\rangle \langle u_{j}| + \frac{\Gamma_{\textrm{Re}, -, \omega_{2}}}{\Gamma_{\textrm{Re}, +, \omega_{2}}+\Gamma_{\textrm{Re}, -, \omega_{2}}} |Zu_{j}\rangle\langle Zu_{j}| \Big)
\end{eqnarray}  But $\rho$ is extremal, then there is only one non-trivial summand in (\ref{sum-rho}). This finishes the proof.
\end{proof}

\begin{corollary}
All invariant states of modified AKV's model are detailed balance.
\end{corollary}
\begin{proof} Let $\rho$ be an invariant state. Clearly the part of $\rho$ supported on $W_{D}$ satisfies a detailed balance condition, so that we can assume that $\rho$ is supported on $V\cap W_{D}^{\perp}$. Recalling that $\{e_{0}, e_{1}, \psi, \psi'\}\subset \ker \rho$ if $\rho$ is an invariant state supported on $V\cap W_{D}^{\perp}$ and $|Z|\psi = \psi$. By direct computations one can see that $\rho D_{\omega_{1}}=0=D_{\omega_{1}} \rho$, and $\rho D_{\omega_{3}}=0=D_{\omega_{3}} \rho$. Moreover, 
\begin{eqnarray*}
\rho D_{\omega_{2}} = \frac{\Gamma_{\textrm{Re}, -, \omega_{2}}}{\Gamma_{\textrm{Re}, +, \omega_{2}}} Z\rho, \; \; \; \textrm{and} \; \; \; D_{\omega_{2}} \rho = Z\rho
\end{eqnarray*} hence, $\rho D_{\omega_{2}} = \frac{\Gamma_{\textrm{Re}, -, \omega_{2}}}{\Gamma_{\textrm{Re}, +, \omega_{2}}} D_{\omega_{2}} \rho$. In view of our Definition (\ref{def-db}) we can conclude that $\rho$ is a detailed balance state.
\end{proof}

\subsection{Fast recurrent subspace}
The fast recurrent subspace of a quantum Markov semigroup is defined as 
\[{\mathcal R} := \sup \{s(\rho): s(\rho) \; \, \textrm{is the support of an invariant state}\}\] 
In this section we will characterize the fast recurrent subspace of the modified AKV's model. 

\begin{lemma}\label{projections} 
With the same notations as above, we have that 
\begin{itemize}
\item[(i)] $P_{W_{D}}= P_{0} + \big(P_{2} - |Z|\big)$
\item[(ii)] $P_{W_{D}^{\perp}}= P_{1} + |Z| + P_{3}$
\item[(iii)] $W_{D}\setminus\{e_{0}\}\subset V$
\item[(iv)] $V\cap W_{D}^{\perp} = \textrm{Im}|Z|\cap \{\psi, Z^{*}\psi'\}^{\perp} \oplus \textrm{Im}P_{3}\cap \{\psi', Z\psi\}^{\perp}$
\end{itemize}
\end{lemma}
\begin{proof} To prove $(i)$ it suffices to prove the identity 
$W_{D} = \{u_{0} e_{0} + \big(P_{2} - |Z|\big) v \, : \, v\in\initsp, \, \, u_{0}\in{\mathbb C}\}$. If $u=u_{0} e_{0} + \big(P_{2} - |Z|\big) v$ with  $v\in\initsp$, $u_{0}\in{\mathbb C}$, clearly $u\perp e_{1}$ and $u\perp \psi'$ since $\psi'\in\textrm{Im}P_{3}$. After simple computations we get   
\[\langle \psi, u \rangle = \langle \psi, (P_{2}- |Z|)v \rangle = \langle (P_{2}- |Z|)\psi, v \rangle =0\] since $P_{2}\psi= \psi = |Z|\psi$, hence $u\perp \psi$. Moreover, we have that $|Z| u = |Z| (P_{2} - |Z|)v=0$, then $u\in \ker |Z|$ and clearly$Z^{*}u = 0$. This proves that $\{u_{0} e_{0} + \big(P_{2} - |Z|\big) v \, : \, v\in\initsp, \, \, u_{0}\in{\mathbb C}\}\subset W_{D}$. To prove the oposite inequality observe that any $u\in\initsp$ may be written in the form $u= P_{0}u + P_{1}u + P_{2}u + P_{3}u$, but if in addition $u\in W_{D}$, then $P_{1}u=0$. Condition $u\in \ker|Z|, \; u\neq P_{0}u$ and identity $P_{2}= \textrm{Im}|Z| \oplus \ker|Z|$, imply that $u=(P_{2}- |Z|) v$ for some $v\in \initsp$ and, hence, $P_{3}u=0$. This proves $(i)$.

Item $(ii)$ is an immediate consequence of $(i)$. If $u\in W_{D}\setminus\{e_{0}\}$, let us say $u= (P_{2} - |Z|)v, \; v\in \initsp,$ then clearly $u\perp \{e_{0}, e_{1}\}$ and using that $P_{2}\psi= \psi = |Z|\psi$, we get 
\[\langle \psi, u\rangle = \langle \psi, (P_{2}-|Z|)v\rangle = \langle (P_{2} - |Z|)\psi, u\rangle =0\] i.e., $u\perp \psi$. Similarly, identities $P_{2}Z^{*}\psi'= Z^{*}\psi' = |Z|Z^{*}\psi'$, imply that $u\perp Z^{*}\psi'$. Moreover, $u\in\{\psi', Z\psi\}^{\perp}$, since $\psi', Z\psi \in \textrm{Im}P_{3}$. This proves that $W_{D}\setminus\{e_{0}\}\subset V$.

Now, using $(ii)$ we have that 
\begin{eqnarray}
\begin{aligned}
V \cap W_{D}^{\perp} = &\{e_{0}, e_{1}, \psi, \psi', Z \psi, Z^{*} \psi' \}^{\perp} \cap \big( P_{1}+|Z|+ P_{3} \big)(\initsp) \\  = &\{\psi, \psi', Z\psi, Z^{*}\psi' \}^{\perp} \cap \big(|Z|+ P_{3}\big)(\initsp) \\  = &\textrm{Im}|Z| \cap\{\psi, Z^{*}\psi'\}^{\perp} \oplus \textrm{Im}P_{3}\cap\{\psi', Z\psi\}^{\perp}
\end{aligned}
\end{eqnarray} This proves $(iv)$ and finishes the proof of the lemma.
\end{proof}
\begin{corollary}\label{inv-notin-comm}
If $M<N-1$, there exists invariant states of the modified AKV's model in $Ann(D)\setminus\{H\}'$.
\end{corollary}
\begin{proof}
Let $u=u_{0} e_{0} + (P_{2}-|Z|)v\in W_{D}$, with $u_{0}\in{\mathbb C}$ and $v\in\initsp$. Then we have,
\begin{eqnarray}
\begin{aligned}
P_{0} |u\rangle\langle u| -|u\rangle\langle u| P_{0} = u_{0} |(P_{2}-|Z|)v\rangle\langle e_{0}| - \bar{u}_{0} |e_{0}\rangle\langle (P_{2}-|Z|)v| 
\end{aligned}
\end{eqnarray} Therefore, the invariant state $|u\rangle\langle u|$ does not belong to $\{H\}'$ if $u_{0}\neq 0$, $v\neq 0$ and $|Z| < P_{2}$ (equivalently $M < N-1$, see Proposition 2.2 in Ref.\cite{gggq}). Nevertheless due to result of Theorem \ref{inv-WD}, the state belongs to the annihilator $Ann(D)$.
\end{proof}

\begin{remark}
Notice that in contradistinction with Ref.\cite{gggq}, in the case when $|Z| < P_{2}$ (equivalently, $M < N-1$) we have \textit{found} invariant states outside $\{H\}'$, but in $Ann(D)$. 
\end{remark}

\begin{proposition}\label{inv-faithful} 
\begin{itemize}
\item[(i)] The restrictions of the maps $Z: \textrm{Im}|Z| \cap\{\psi, Z^{*}\psi'\}^{\perp}\mapsto \textrm{Im}P_{3} \cap\{\psi', Z\psi\}^{\perp}$ and $Z^{*}: \textrm{Im}P_{3} \cap\{\psi', Z\psi\}^{\perp}\mapsto\textrm{Im}|Z| \cap\{\psi, Z^{*}\psi'\}^{\perp}$ are unitaries. Consequently $\textrm{dim}\Big(\textrm{Im}|Z| \cap\{\psi, Z^{*}\psi'\}^{\perp}\Big)= \textrm{dim}\Big( \textrm{Im}P_{3} \cap\{\psi', Z\psi\}^{\perp}\Big)$.
 
\item[(ii)] The state \[\rho = \frac{1}{\textrm{tr}(q)}\Big(\frac{\Gamma_{\textrm{Re}, +, \omega_{2}}}{\Gamma_{\textrm{Re}, +, \omega_{2}}+\Gamma_{\textrm{Re}, -, \omega_{2}}}  q + \frac{\Gamma_{\textrm{Re}, -, \omega_{2}}}{\Gamma_{\textrm{Re}, +, \omega_{2}}+\Gamma_{\textrm{Re}, -, \omega_{2}}} Z q Z^{*}\Big)\] where $q= |Z| P_{\{\psi, Z^{*}\psi'\}^{\perp}}$ is the orthogonal projection on $\textrm{Im}|Z| \cap V$, is invariant for the modified AKV's model, with support projection $P_{W_{D}^{\perp}\cap V} = q + Z q Z^{*}$. Hence, $\rho$ is faithful as an state in the hereditary subalgebra $P_{W_{D}^{\perp}\cap V}{\mathcal B}(\initsp)P_{W_{D}^{\perp}\cap V}$.
\end{itemize}
\end{proposition}
\begin{proof}
Notice that $Z$ maps $\textrm{Im}|Z| \cap\{\psi, Z^{*}\psi'\}^{\perp}$ into $\textrm{Im}P_{3} \cap\{\psi', Z\psi\}^{\perp}$. Indeed, if $u\in \textrm{Im}|Z| \cap\{\psi, Z^{*}\psi'\}^{\perp}$ then $P_{3}Zu= Zu$, $\langle \psi', Zu\rangle = \langle Z^{*} \psi', u\rangle=0$ and $\langle Z\psi, Zu\rangle= \langle |Z| \psi, u\rangle = \langle \psi, u\rangle=0$, hence $Z u\in\textrm{Im}P_{3} \cap\{\psi', Z\psi\}^{\perp}$. Moreover, $\langle Zu, Zv\rangle= \langle |Z| u, v\rangle = \langle u, v\rangle$, meaning that $Z$ is unitary. Analogously one shows that $Z^{*}$ maps $\textrm{Im}P_{3} \cap\{\psi', Z\psi\}^{\perp}$ into $\textrm{Im}|Z| \cap\{\psi, Z^{*}\psi'\}^{\perp}$ and is unitary. This proves $(i)$.  

Let $q$ be the orthogonal projection on the subspace $\textrm{Im}|Z| \cap V$. Notice that $q=|Z|$ on $\textrm{Im}|Z| \cap\{\psi, Z^{*}\psi'\}^{\perp}$. Indeed, since $\psi$ and $Z^{*}\psi'$ are fixed points of $|Z|$, this operator maps $\textrm{Im}|Z| \cap\{\psi, Z^{*}\psi'\}^{\perp}$ into itself and condition $|Z|u=0, \; u\in \textrm{Im}|Z| \cap\{\psi, Z^{*}\psi'\}^{\perp}$, implies that $u= |Z|u=0$, then it is injective. Consequently, the map $Zq Z^{*}$ maps $\textrm{Im}P_{3} \cap\{\psi', Z\psi\}^{\perp}$ into itself and it is injective since $Z q Z^{*} u=0$ and $u\in \textrm{Im}P_{3} \cap\{\psi', Z\psi\}^{\perp}$, imply $0=Z|Z|Z^{*}u= P_{3}u=u$. Moreover $ZqZ^{*}$ is a projection since $Z q Z^{*} Z q Z^{*} = Z q |Z| q Z^{*} = ZqZ^{*}$, and for any $u\in P_{3}(\initsp)= ZZ^{*}(\initsp)$ we have that $Z qZ^{*} u = Z q Z^{*} ZZ^{*}u= Zq|Z|Z^{*}u= ZZ^{*}ZZ^{*}u=u $, hence, coincides with the orthogonal projection on $\textrm{Im}P_{3} \cap\{\psi', Z\psi\}^{\perp}$. Consequently, by item $(iv)$ of the above Lemma \ref{projections} we conclude that $P_{W_{D}^{\perp}\cap V} = q + Z q Z^{*}$. The result of the proposition readily follows from Theorem \ref{explicit-inv-states}.
\end{proof}

\begin{corollary}
The fast recurrent subspace of the modified AKV's model is \[{\mathcal R}= W_{D}\oplus W_{D}^{\perp}\cap V\]
\end{corollary}
\begin{proof}
By Theorem \ref{explicit-inv-states}, every invariant state is supported on a subspace of ${\mathcal R}$. The same theorem and Proposition \ref{inv-faithful} imply that there exits an invariant state with support equal to ${\mathcal R}$.
\end{proof}

\subsection{Approach to equilibrium and attraction domains}

The decoherence free sub-algebra ${\mathcal N}({\mathcal T)}$ of the QMS, ${\mathcal T}$, of the AKV's model was studied and characterized in Ref.\cite{agq} It was proved there that ${\mathcal N}({\mathcal T)}\subset{\mathcal F}({\mathcal T})$ where ${\mathcal F}({\mathcal T})$ is the subset of fixed point of ${\mathcal T}$, consequently, ${\mathcal N}({\mathcal T)}={\mathcal F}({\mathcal T})$ since the opposite inequality always holds true. If there exists a faithful invariant state in ${\mathcal B}(\initsp)$, from a result of Frigerio and Verri\cite{af,fv} one can conclude that for any normal state $\eta\in {\mathcal B}(\initsp)$ there exists the limit $\tau = \lim_{t\to\infty}{\mathcal T}_{*t}(\eta)$. Unfortunately, as a consequence of Theorem \ref{explicit-inv-states}, any invariant state $\tau$ for the AKV's model is singular, indeed $V^{\perp}\setminus\{e_{0}\}= \textrm{Span}\{e_{1}, \psi, \psi', Z\psi, Z^{*}\psi'\}\subset \ker(\tau)$. Then, to analyze the approach to equilibrium property and attraction domains, it is necessary to restrict our attention to the evolution on hereditary sub-algebras ${\mathcal A}_{W_{D}} = P_{W_{D}}{\mathcal B}(\initsp)P_{W_{D}}$ and ${\mathcal A}_{W_{D}^{\perp}\cap V} = P_{W_{D}^{\perp}\cap V}{\mathcal B}(\initsp)P_{W_{D}^{\perp}\cap V}$.

\begin{lemma}\label{app-equil} 
Let ${\mathcal T}_{W_{D}^{\perp}\cap V, t}= P_{W_{D}^{\perp}\cap V}{\mathcal T}_{t} P_{W_{D}^{\perp}\cap V}$ be the hereditary semigroup ac\-ting on the subalgebra ${\mathcal A}_{W_{D}^{\perp}\cap V}$. Then ${\mathcal N}\Big({\mathcal T}_{W_{D}^{\perp}\cap V} \Big) \subset {\mathcal F}\Big({\mathcal T}_{W_{D}^{\perp}\cap V}\Big)$.
\end{lemma}
\begin{proof} Let $a$ be a selfadjoint element $a\in {\mathcal A}_{W_{D}^{\perp}\cap V}$, then $a$ is supported on a subspace of $V$, meaning that $\{e_{1}, \psi, \psi', Z\psi, Z^{*}\psi'\}\subset \ker a$. Also $e_{0}\in\ker a$ since $a$ is supported on $W_{D}^{\perp}$. Then $a$ commutes with $P_{0}$ and $P_{1}$. If moreover $a\in {\mathcal N}\Big({\mathcal T}_{W_{D}^{\perp}\cap V} \Big)$ then $a\in {\mathcal C}_{0}'$, since by Theorem 3.2 in Ref.\cite{adffrr} ${\mathcal N}({\mathcal T}_{W_{D}^{\perp}\cap V})\subset {\mathcal C}_{0}'$, where 
\begin{eqnarray}\label{pre-1} 
\begin{aligned}
{\mathcal C}_{0} &= \{D_{\omega_{\ell}}, D_{\omega_{\ell}}^{*} : \; 1\leq \ell \leq 3\} \\ &= \{|\psi\rangle \langle e_{1}|, |e_{1}\rangle\langle \psi|, Z, Z^{*}, |e_{0}\rangle \langle \psi'|, |\psi'\rangle \langle e_{0}| \}
\end{aligned}
\end{eqnarray} consists of the Kraus operators of the model along with their adjoints, restricted to ${\mathcal A}_{W_{D}^{\perp} \cap V}$. In particular $a$ commutes with $Z$ and $Z^{*}$, hence $a\big(\ker|Z|\big)\subset\ker|Z|$ and $a\big(\textrm{Im}|Z|\big)\subset\textrm{Im}|Z|$. Then $a$ commutes with $P_{2}= \textrm{Im}|Z| \oplus \ker |Z|$, an consequently, it commutes also with $P_{3}= I-P_{0} -P_{1} - P_{2}$. In conclusion $a$ commutes with $H$. 

Now, condition $\{e_{0}, e_{1}, \psi, \psi', Z\psi, Z^{*}\psi'\}\subset \ker a$ implies that $a\in \ker {\mathcal L}_{\omega_{1}}\cap \ker{\mathcal L}_{\omega_{3}}$. Moreover since $a$ commutes with $P_{3}$ and with $|Z|$, $[H_{\omega_{2}}, a]=0$. Commutation of $a$ with $Z, Z^{*}$ implies that $a\in \ker{\mathcal L}_{\omega_{2}}$. Consequently $a$ is a fixed point of ${\mathcal T}_{W_{D}^{\perp}\cap V}$. Since any operator is a linear combination of four selfadjoint elements, the same result holds true for any element $a\in {\mathcal A}_{W_{D}^{\perp}\cap V}$. This finishes the proof. \\
\end{proof}

\begin{corollary}\label{limit-state}
For any initial state $\eta\in {\mathcal A}_{W_{D}^{\perp}\cap V}$ there exists an invariant state $\eta_{\infty}= \lim_{t\to\infty} {\mathcal T}_{W_{D}^{\perp}\cap V, *t}(\eta)\in {\mathcal A}_{W_{D}^{\perp}\cap V}$. 
\end{corollary}
\begin{proof}
This is a consequence of Lemma \ref{app-equil}, Proposition \ref{inv-faithful} and the previously cited result of Frigerio and Verri\cite{af,fv}.
\end{proof}
 Now we shall give a full description of the evolution of any state supported on $W_{D}^{\perp} \cap V$. To do so, we start by describing explicitly the evolution of states supported either on $\textrm{Im}|Z|\cap \{\psi, Z^{*}\psi'\}^{\perp}$ or on $\textrm{Im}P_{3}\cap \{\psi', Z\psi\}^{\perp}$. 

\begin{lemma}\label{portions}
\begin{itemize}
\item[(i)] Let $\sigma$ be a state supported on $\textrm{Im}|Z|\cap\{\psi, Z^{*}\psi'\}^{\perp}$. Then ${\mathcal T}_{*t}(\sigma)$ converges to 
\[\frac{\Gamma_{\textrm{Re}, +, \omega_{2}}}{\Gamma_{\textrm{Re}, +, \omega_{2}} + \Gamma_{\textrm{Re}, -, \omega_{2}}} \sigma +   \frac{\Gamma_{\textrm{Re}, -, \omega_{2}}}{\Gamma_{\textrm{Re}, +, \omega_{2}} + \Gamma_{\textrm{Re}, -, \omega_{2}}} \; Z\sigma Z^{*}\]  as $t\to\infty$.
\item[(ii)] Let $\eta$ be a state supported on $\textrm{Im}P_{3}\cap\{\psi', Z\psi\}^{\perp}$. Then ${\mathcal T}_{*t} (\eta)$ converges to 
\[\frac{\Gamma_{\textrm{Re}, +, \omega_{2}}}{\Gamma_{\textrm{Re}, +, \omega_{2}} + \Gamma_{\textrm{Re}, -, \omega_{2}}} Z^{*}\eta Z +   \frac{\Gamma_{\textrm{Re}, -, \omega_{2}}}{\Gamma_{\textrm{Re}, +, \omega_{2}} + \Gamma_{\textrm{Re}, -, \omega_{2}}} \; \eta\]  as $t\to\infty$
\end{itemize} 
\end{lemma}
\begin{proof}
For any observable $x\in{\mathcal B}(\initsp)$ define $x_{t}= \textrm{tr}\big(x{\mathcal T}_{*t}(\sigma)\big)$ and $y_{t}= \textrm{tr}\big(x{\mathcal T}_{*t}(Z \sigma Z^{*})\big)$. Then after derivation with respect to $t$ we get that 
\begin{eqnarray}\label{system}
\begin{aligned}
\dot{x}_{t} = & -a x_{t} + a y_{t} \\ 
\dot{y}_{t} = & \; b x_{t} - b y_{t} 
\end{aligned}
\end{eqnarray} with $a= 2 (N-1)\Gamma_{\textrm{Re}, -, \omega_{2}}$ and $b=2 (N-1)\Gamma_{\textrm{Re}, +, \omega_{2}}$. Therefore $(x_{t}, y_{t})$ is a classical random walk with $Q$-matrix 
\begin{eqnarray}
Q= \begin{array}{cc}\left(\begin{array}{cc} -a & a \\b & -b 
\end{array}\right)\end{array}
\end{eqnarray} and stationary probability distribution $\pi=\big(\frac{b}{a+b}, \frac{a}{a+b} \big)$. Indeed, after solving system (\ref{system}) one gets, 
\[x_{t}= \textrm{tr} \Big(x \big(\frac{b + a e^{-t(a+b)}}{a+b} \sigma + \frac{a - a e^{-t(a+b)}}{a+b} Z\sigma Z^{*}\big) \Big)\] which implies that \[{\mathcal T}_{*t}(\sigma)= \frac{b + a e^{-t(a+b)}}{a+b} \sigma + \frac{a - a e^{-t(a+b)}}{a+b} Z\sigma Z^{*}\mapsto \frac{b}{a+b} \sigma + \frac{a}{a+b} Z\sigma Z^{*}\] as $t\to\infty$. This proves $(i)$.

To proof $(ii)$ take $\sigma = Z^{*}\eta Z$ and apply $(i)$. This finishes the proof.
\end{proof}

\begin{theorem}\label{approach-equilibrium}
Let $\eta$ be an initial state supported on $W_{D}^{\perp}\cap V$. Then 
\begin{eqnarray}
\begin{aligned}
\lim_{t\to\infty} {\mathcal T}_{*t}(\eta) =  \; \frac{\Gamma_{\textrm{Re}, +, \omega_{2}}}{\Gamma_{\textrm{Re}, +, \omega_{2}} + \Gamma_{\textrm{Re}, -, \omega_{2}}} \sigma + \; \frac{\Gamma_{\textrm{Re}, -, \omega_{2}}}{\Gamma_{\textrm{Re}, +, \omega_{2}} + \Gamma_{\textrm{Re}, -, \omega_{2}}} Z \sigma Z^{*}
\end{aligned}
\end{eqnarray} with $\sigma = \frac{1}{\textrm{tr}(q\eta)}q \eta q + \frac{1}{\textrm{tr}(P_{3}\eta)} Z^{*}\eta Z$.
\end{theorem} 
\begin{proof}
By Corollary \ref{limit-state} there exists an invariant state 
\begin{eqnarray}\label{rho-infty}
\eta_{\infty}= \frac{\Gamma_{\textrm{Re}, +, \omega_{2}}}{\Gamma_{\textrm{Re}, +, \omega_{2}} + \Gamma_{\textrm{Re}, -, \omega_{2}}} \sigma + \frac{\Gamma_{\textrm{Re}, -, \omega_{2}}}{\Gamma_{\textrm{Re}, +, \omega_{2}} + \Gamma_{\textrm{Re}, -, \omega_{2}}}Z\sigma Z^{*}
\end{eqnarray} where $\sigma$ is a state supported on $\textrm{Im}|Z|\cap\{\psi, Z^{*}\psi'\}^{\perp}$, such that $\lim_{t\to\infty}{\mathcal T}_{*t}(\eta)=\eta_{\infty}$. 

Every initial state $\eta$ supported on $W_{D}^{\perp}\cap V$ may be written in the form
\begin{eqnarray}\label{initial}
\eta= (q + ZqZ^{*}) \eta (q + ZqZ^{*})= q\eta q+ ZqZ^{*}\eta ZqZ^{*} + q\eta ZqZ^{*} + ZqZ^{*}\eta q
\end{eqnarray} The state $\frac{1}{\textrm{tr}(q \eta)} q\eta q$ is supported on $\textrm{Im}|Z|\cap\{\psi, Z^{*}\psi'\}^{\perp}$ and due to the above Lemma \ref{portions}, it is driven by the semigroup to the limit state 
\begin{eqnarray}\label{portion-1}
\frac{1}{\textrm{tr}(q \eta)} \Big(\frac{\Gamma_{\textrm{Re}, +, \omega_{2}}}{\Gamma_{\textrm{Re}, +, \omega_{2}} + \Gamma_{\textrm{Re}, -, \omega_{2}}}  q\eta q + \frac{\Gamma_{\textrm{Re}, -, \omega_{2}}}{\Gamma_{\textrm{Re}, +, \omega_{2}} + \Gamma_{\textrm{Re}, -, \omega_{2}}}  Zq\eta q Z^{*}\Big)
\end{eqnarray} as $t\to\infty$. Analogously, due to the same Lemma \ref{portions}, the state $\frac{1}{\textrm{tr}(P_{3}\eta)} ZqZ^{*}\eta ZqZ^{*}= \frac{1}{\textrm{tr}(P_{3}\eta)}P_{3}\eta P_{3}$ is supported on $\textrm{Im}P_{3}\cap\{\psi', Z\psi\}^{\perp}$ and 
\begin{eqnarray}\label{portion-2}
\begin{aligned}
\lim_{t\to\infty} {\mathcal T}_{*t}\big(\frac{1}{\textrm{tr}(P_{3}\eta)}P_{3}\eta P_{3}\big) = & 
\frac{1}{\textrm{tr}(P_{3}\eta)} \Big(\frac{\Gamma_{\textrm{Re}, +, \omega_{2}}}{\Gamma_{\textrm{Re}, +, \omega_{2}} + \Gamma_{\textrm{Re}, -, \omega_{2}}}  Z^{*}P_{3}\eta P_{3}Z \\ + & \frac{\Gamma_{\textrm{Re}, -, \omega_{2}}}{\Gamma_{\textrm{Re}, +, \omega_{2}} + \Gamma_{\textrm{Re}, -, \omega_{2}}}  P_{3}\eta P_{3}\Big)
\end{aligned}
\end{eqnarray}
Comparing (\ref{rho-infty}) with (\ref{portion-1}) and (\ref{portion-2}), we conclude that 
\[\lim_{t\to\infty}{\mathcal T}_{*t}\big(q\eta ZqZ^{*} + ZqZ^{*}\eta q\big)=0,\] 
\[\sigma = \frac{1}{\textrm{tr}(q \eta)} q\eta q + \frac{1}{\textrm{tr}(P_{3} \eta)} Z^{*}P_{3}\eta P_{3}Z\] and consequently, 
\[Z\sigma Z^{*}= \frac{1}{\textrm{tr}(q 
\eta)} Zq\eta q Z^{*}+ \frac{1}{\textrm{tr}(P_{3} \eta)} P_{3}\eta P_{3}\] This finishes the proof.
 \end{proof}
 
\begin{corollary}
The attraction domain of the invariant state \[\rho = \frac{\Gamma_{\textrm{Re}, +, \omega_{2}}}{\Gamma_{\textrm{Re}, +, \omega_{2}} + \Gamma_{\textrm{Re}, -, \omega_{2}}} \sigma +   \frac{\Gamma_{\textrm{Re}, -, \omega_{2}}}{\Gamma_{\textrm{Re}, +, \omega_{2}} + \Gamma_{\textrm{Re}, -, \omega_{2}}} \; Z\sigma Z^{*}\] where $\sigma$ is a state supported on $\textrm{Im}|Z|\cap\{\psi, Z^{*}\psi'\}^{\perp}$, consists of all those initial states $\eta$ such that 
\[\sigma = \frac{1}{\textrm{tr}(q \eta)} q\eta q + \frac{1}{\textrm{tr}(P_{3} \eta)} Z^{*}P_{3}\eta P_{3}Z\] 
\end{corollary}
\begin{proof}
Immediate from the above Theorem.
\end{proof}

\begin{remark}\label{transport}\textbf{(Transport of states)}
Due to the result of Theorem \ref{limit-state} the total probability of an initial state $\eta$ is redistributed in the limit as $t\to\infty$; so that the probability in the portion of the final state $\eta_{\infty}$ supported on $\textrm{Im}|Z|\cap \{\psi, Z^{*}\psi'\}^{\perp}$ is \[\textrm{tr}\big(q \eta_{\infty}\big)= \frac{\Gamma_{\textrm{Re}, +, \omega_{2}}}{\Gamma_{\textrm{Re}, +, \omega_{2}} + \Gamma_{\textrm{Re}, -, \omega_{2}}}\] and the probability in the portion of the final state $\eta_{\infty}$ supported on $\textrm{Im}P_{3}\cap \{\psi', Z\psi\}^{\perp}$ is \[\textrm{tr}\big(P_{3} \eta_{\infty}\big)= \frac{\Gamma_{\textrm{Re}, -, \omega_{2}}}{\Gamma_{\textrm{Re}, +, \omega_{2}} + \Gamma_{\textrm{Re}, -, \omega_{2}}}\] 

By adjusting the values of the $\Gamma$'s, so that $\frac{\Gamma_{\textrm{Re}, +, \omega_{2}}} {\Gamma_{\textrm{Re}, -, \omega_{2}}}=e^{-\beta(\omega_{2})\omega_{2}}\to 0$ (equivalently $\beta(\omega_{2})\to\infty$), one can ``transport'' any initial state $\eta$ to a limit state $\eta_{\infty}$ concentrated (but non-necessarily supported) on the subspace $\textrm{Im}P_{3}\cap \{\psi', Z\psi\}^{\perp}$. In particular any initial state supported on $\textrm{Im}|Z|\cap \{\psi, Z^{*}\psi'\}^{\perp}$ can be transported, as $t\to\infty$, to a limit state concentrated on the subspace $\textrm{Im}P_{3}\cap \{\psi', Z\psi\}^{\perp}$. In this sense our modified AKV's model allows us to ``transport'' probability mass from any initial state to a state concentrated on $\textrm{Im}P_{3}\cap \{\psi', Z\psi\}^{\perp}$.  
\end{remark}

\begin{remark}\label{energy-gain}\textbf{(Energy gain)}
A simple computation shows that if at $t=0$ the system is in any initial state $\rho$ supported on the subspace $\textrm{Im}|Z|\cap \{\psi, Z^{*}\psi'\}^{\perp}$, then as $t\to\infty$ the energy of the system increases and the \textit{energy gain} during the process is given by \[\textrm{tr}(\rho_{\infty} H_{eff})- \textrm{tr}(\rho H_{eff})= (N-1) \frac{\Gamma_{\textrm{Re}, -, \omega_{2}}}{\Gamma_{\textrm{Re}, +, \omega_{2}} + \Gamma_{\textrm{Re}, -, \omega_{2}}} (\Gamma_{\textrm{Im},+,\omega_{2}}+ \Gamma_{\textrm{Im},-,\omega_{2}})\] which is independent of $\rho$. This is an indication that degenerate open systems (with a degenerate reference Hamiltonian) seems to be appropriate for modeling effective quantum energy transfer in photosynthesis, see Ref.\cite{Scholes-et-al} and the reference therein. We will continue the analysis of energy transfer when the initial state of the system is arbitrary in forthcoming work.
\end{remark}
 
\section*{Acknowledgment}
The authors are grateful to Franco Fagnola who suggested to work in this problem and gave valuable insight. The financial support from CONACYT-Mexico (Grant 221873) and PRODEP Red de Analisis Italia-UAM, is gratefully acknow\-led\-ged

\end{document}